\documentclass[12pt]{article}
\usepackage[a4paper,margin=1.2in]{geometry}

\usepackage[plain,noend]{algorithm2e}
\usepackage{amsmath}
\usepackage{amssymb}
\usepackage{bm}
\usepackage{color}
\usepackage{commath}
\usepackage{enumerate}
\usepackage{graphicx}
\usepackage[utf8]{inputenc}
\usepackage{mathtools}
\usepackage{natbib}
\usepackage{placeins}
\usepackage{subcaption}
\usepackage{verbatim}
\usepackage{xcolor}
\usepackage{xr}

\allowdisplaybreaks

\makeatletter
\renewcommand{\algocf@captiontext}[2]{#1\algocf@typo. \AlCapFnt{}#2} 
\def\@algocf@capt@plain{top}
\renewcommand{\algocf@makecaption}[2]{%
  \addtolength{\hsize}{\algomargin}%
  \sbox\@tempboxa{\algocf@captiontext{#1}{#2}}%
  \ifdim\wd\@tempboxa >\hsize
    \hskip .5\algomargin%
    \parbox[t]{\hsize}{\algocf@captiontext{#1}{#2}}
  \else%
    \global\@minipagefalse%
    \hbox to\hsize{\box\@tempboxa}
  \fi%
  \addtolength{\hsize}{-\algomargin}%
}
\makeatother

\DeclareUnicodeCharacter{00A0}{ } 

\DeclareMathOperator*{\argmin}{arg\,min}

\renewcommand{\b}[1]{\bm{#1}} 

\usepackage{amsmath}
\usepackage{amssymb}
\usepackage{amsthm}

\DeclareUnicodeCharacter{00A0}{ } 

\usepackage{natbib} 
\newtheorem{lemma}{Lemma}
\newtheorem{theorem}{Theorem}

\newtheorem{corollary}{Corollary}
\newtheorem{proposition}{Proposition}

\begin{document}

\title{\textbf{Semi-Exact Control Functionals From Sard's Method}}
\author{Leah F. South$^1$, Toni Karvonen$^2$, Chris Nemeth$^3$, \\
Mark Girolami$^{4,2}$, Chris. J. Oates$^{5,2}$ \\ \small
$^1$Queensland University of Technology, Australia \\ \small
$^2$The Alan Turing Institute, UK \\ \small
$^3$Lancaster University, UK \\ \small
$^4$University of Cambridge, UK \\ \small
$^5$Newcastle University, UK }
\maketitle

\begin{abstract}
The numerical approximation of posterior expected quantities of interest is considered.  
A novel control variate technique is proposed for post-processing of Markov chain Monte Carlo output, based both on Stein's method and an approach to numerical integration due to Sard. 
The resulting estimators are proven to be polynomially exact in the Gaussian context, while empirical results suggest the estimators approximate a Gaussian cubature method near the Bernstein-von-Mises limit. 
The main theoretical result establishes a bias-correction property in settings where the Markov chain does not leave the posterior invariant.
Empirical results are presented across a selection of Bayesian inference tasks. 
All methods used in this paper are available in the \texttt{R} package \texttt{ZVCV}.

\vspace{5pt}
\end{abstract}

\section{Introduction}
This paper focuses on the numerical approximation of integrals of the form
$$
I(f) = \int f(\b{x}) p(\b{x}) \mathrm{d} \b{x} ,
$$
where $f$ is a function of interest and $p$ is a positive and continuously differentiable probability density on $\mathbb{R}^d$, under the restriction that $p$ and its gradient can only be evaluated pointwise up to an intractable normalisation constant.
The standard approach to computing $I(f)$ in this context is to simulate the first $n$ steps of a $p$-invariant Markov chain $(\b{x}^{(i)})_{i=1}^\infty$, possibly after an initial burn-in period, and to take the average along the sample path as an approximation to the integral:
\begin{equation}\label{eqn:MC}
I(f) \approx I_{\text{MC}}(f) = \frac{1}{n}\sum_{i=1}^n f(\b{x}^{(i)}).
\end{equation}
See Chapters 6--10 of \cite{Robert2013} for background.
In this paper $\mathbb{E}$, $\mathbb{V}$ and $\mathbb{C}$ respectively denote expectation, variance and covariance with respect to the law $\mathbb{P}$ of the Markov chain.
Under regularity conditions on $p$ that ensure the Markov chain $(\b{x}^{(i)})_{i=1}^\infty$ is aperiodic, irreducible and reversible, the convergence of $I_{\text{MC}}(f)$ to $I(f)$ as $n \rightarrow \infty$ is described by a central limit theorem 
\begin{equation}
\sqrt{n} (I_{\text{MC}}(f) - I(f)) \rightarrow \mathcal{N}(0,\sigma(f)^2) \label{eqn: clt}
\end{equation}
where convergence occurs in distribution and, if the chain starts in stationarity,
\begin{equation*}
\sigma(f)^2 = \mathbb{V}[f(\b{x}^{(1)})] + 2 \sum_{i=2}^\infty \mathbb{C}[f(\b{x}^{(1)}), f(\b{x}^{(i)})]
\end{equation*}
is the asymptotic variance of $f$ along the sample path.
See Theorem 4.7.7 of \cite{Robert2013} and more generally \cite{meyn2012markov} for theoretical background.
Note that for all but the most trivial function $f$ we have $\sigma(f)^2 > 0$ and hence, to achieve an approximation error of $O_P(\epsilon)$, a potentially large number $O(\epsilon^{-2})$ of calls to $f$ and $p$ are required. 

One approach to reduce the computational cost is to employ control variates \citep{Hammersley1964,Ripley1987}, which involves finding an approximation $f_n$ to $f$ that can be exactly integrated under $p$, such that $\sigma(f - f_n)^2 \ll \sigma(f)^2$. 
Given a choice of $f_n$, the standard estimator \eqref{eqn:MC} is replaced with 
\begin{equation}\label{eqn:CV}
I_{\text{CV}}(f) =   \frac{1}{n} \sum_{i=1}^n [f(\b{x}^{(i)}) - f_n(\b{x}^{(i)})] + \underbrace{ \int f_n(\b{x})p(\b{x}) \mathrm{d}\b{x} }_{(*)} ,
\end{equation}
where $(*)$ is exactly computed.
This last requirement makes it challenging to develop control variates for general use, particularly in Bayesian statistics where often the density $p$ can only be accessed in a form that is un-normalised.
In the Bayesian context, \citet{Assaraf1999,Mira2013} and \citet{Oates2017} addressed this challenge by using $f_n = c_n + \mathcal{L}g_n$ where $c_n \in \mathbb{R}$,
$g_n$ is a user-chosen parametric or non-parametric function and $\mathcal{L}$ is an operator, for example the Langevin Stein operator \citep{Stein1972,Gorham2015}, that depends on $p$ through its gradient and satisfies $\int (\mathcal{L} g_n)(\b{x})p(\b{x}) \mathrm{d}\b{x} = 0$ under regularity conditions (see Lemma \ref{lemma:assum}). 
Convergence of $I_{\text{CV}}(f)$ to $I(f)$ has been studied under (strong) regularity conditions and, in particular (i) if $g_n$ is chosen parametrically, then in general $\lim\inf \sigma(f - f_n)^2 > 0$ so that, even if asymptotic variance is reduced, convergence rates are unaffected; (ii) if $g_n$ is chosen in an appropriate non-parametric manner then $\lim \sup \sigma(f - f_n)^2 = 0$ and a smaller number $O(\epsilon^{-2 + \delta})$, $0 < \delta < 2$, of calls to $f$, $p$ and its gradient are required to achieve an approximation error of $O_P(\epsilon)$ for the integral \citep[see][]{Oates2019,Mijatovic2018,Barp2018,Belomestny2017,Belomestny2019a,Belomestny2019}.
In the parametric case $\mathcal{L} g_n$ is called a \emph{control variate} while in the non-parametric case it is called a \emph{control functional}.

Practical parametric approaches to the choice of $g_n$ have been well-studied in the Bayesian context, typically based on polynomial regression models \citep{Assaraf1999,Mira2013,Papamarkou2014,Oates2016,Brosse2019}, but neural networks have also been proposed recently \citep{Zhu2018,Si2020}.
In particular, existing control variates based on polynomial regression have the attractive property of being \emph{semi-exact}, meaning that there is a well-characterized set of functions $f \in \mathcal{F}$ for which $f_n$ can be shown to exactly equal $f$ after a finite number of samples $n$ have been obtained.
For the control variates of \cite{Assaraf1999} and \cite{Mira2013} the set $\mathcal{F}$ contains certain low order polynomials when $p$ is a Gaussian distribution on $\mathbb{R}^d$.
Those authors term their control variates ``zero variance'', but we prefer the term ``semi-exact'' since a general integrand $f$ will not be an element of $\mathcal{F}$. 
Regardless of terminology, semi-exactness of the control variate is an appealing property because it implies that the approximation $I_{\text{CV}}(f)$ to $I(f)$ is exact on $\mathcal{F}$.
Intuitively, the performance of the control variate method is related to the richness of the set $\mathcal{F}$ on which it is exact.
For example, polynomial exactness of cubature rules is used to establish their high order convergence rates using a Taylor expansion argument \citep[e.g.][Chapter~8]{Hildebrand1987}.

The development of non-parametric approaches to the choice of $g_n$ has to-date focused on kernel methods \citep{Oates2017,Barp2018}, piecewise constant approximations \citep{Mijatovic2018} and non-linear approximations based on selecting basis functions from a dictionary \citep{Belomestny2017,South2019}. 
Theoretical analysis of non-parametric control variates was provided in the papers cited above, but compared to parametric methods, practical implementations of non-parametric methods are less well-developed.

In this paper we propose a semi-exact control functional method.
This constitutes the ``best of both worlds'', where at small $n$ the semi-exactness property promotes stability and robustness of the estimator $I_{\text{CV}}(f)$, while at large $n$ the non-parametric regression component can be used to accelerate the convergence of $I_{\text{CV}}(f)$ to $I(f)$. 
In particular we argue that, in the Bernstein-von-Mises limit, the set $\mathcal{F}$ on which our method is exact is precisely the set of low order polynomials, so that our method can be considered as an approximately polynomially-exact cubature rule developed for the Bayesian context.
Furthermore, we establish a bias-correcting property, which guarantees the approximations produced using our method are consistent in certain settings where the Markov chain is not $p$-invariant. 

Our motivation comes from the approach to numerical integration due to \cite{Sard1949}. 
Many numerical integration methods are based on constructing an approximation $f_n$ to the integrand $f$ that can be exactly integrated.
In this case the integral $I(f)$ is approximated using $(*)$ in \eqref{eqn:CV}. 
In Gaussian and related cubatures, the function $f_n$ is chosen in such a way that polynomial exactness is guaranteed \citep[Section~1.4]{Gautschi2004}.
On the other hand, in kernel cubature and related approaches, $f_n$ is an element of a reproducing kernel Hilbert space chosen such that an error criterion is minimised \citep{Larkin1970}.
The contribution of Sard was to combine these two concepts in numerical integration by choosing $f_n$ to enforce exactness on a low-dimensional space $\mathcal{F}$ of functions and use the remaining degrees of freedom to find a minimum-norm interpolant to the integrand.

The remainder of the paper is structured as follows:
Section~\ref{subsec: saard's method} recalls Sard's approach to integration and Section~\ref{subsec: stein operators} how Stein operators can be used to construct a control functional.
The proposed semi-exact control functional estimator $I_{\text{SECF}}$ is presented in Section~\ref{ssec:proposedBSS} and its polynomial exactness in the Bernstein-von-Mises limit is discussed in Section~\ref{subsec: PE in BvM}.
A closed-form expression for the resulting estimator $I_{\text{CV}}$ is provided in Section~\ref{subsec: computation}. 
The statistical and computational efficiency of the proposed semi-exact control functional method is compared with that of existing control variates and control functionals using several simulation studies in Section~\ref{sec: Empirical}. 
Practical diagnostics for the proposed method are established in Section~\ref{sec: theory}.
The paper concludes with a discussion in Section~\ref{sec: Discussion}.

\section{Methods} \label{sec: methods}

In this section we provide background details on Sard's method and Stein operators before describing the semi-exact control functional method. 

\subsection{Sard's Method} \label{subsec: saard's method}

Many popular methods for numerical integration are based on either (i) enforcing \emph{exactness} of the integral estimator on a finite-dimensional set of functions $\mathcal{F}$, typically a linear space of polynomials, or on (ii) integration of a \emph{minimum-norm interpolant} selected from an infinite-dimensional set of functions $\mathcal{H}$. 
In each case, the result is a cubature method of the form
\begin{equation} \label{eq:quadrature-generic}
    I_{\text{NI}}(f) = \sum_{i=1}^n w_i f(\b{x}^{(i)}) 
\end{equation}
for weights $\{w_i\}_{i=1}^n \subset \mathbb{R}$ and points $\{\b{x}^{(i)}\}_{i=1}^n \subset \mathbb{R}^d$. Classical examples of methods in the former category are the univariate Gaussian quadrature rules~\citep[Section~1.4]{Gautschi2004}, which are determined by the unique $\{(w_i, \b{x}^{(i)})\}_{i=1}^n \subset \mathbb{R} \times \mathbb{R}^d$ such that $I_{\text{NI}}(f) = I(f)$ whenever $f$ is a polynomial of order at most $2n-1$, and Clenshaw--Curtis rules~\citep{ClenshawCurtis1960}. Methods of the latter category specify a suitable normed space $(\mathcal{H}, \|\cdot\|_{\mathcal{H}})$ of functions, construct an interpolant $f_n \in \mathcal{H}$ such that 
\begin{equation} \label{eq:minimum-norm}
    f_n \in \argmin_{ h \in \mathcal{H}} \big\{ \norm[0]{h}_{\mathcal{H}} \, \colon \, h(\b{x}^{(i)}) = f(\b{x}^{(i)}) \text{ for } i = 1, \ldots, n \big\}
\end{equation}
and use the integral of $f_n$ to approximate the true integral. 
Specific examples include splines~\citep{Wahba1990} and kernel or Gaussian process based methods~\citep{Larkin1970,OHagan1991,Briol2019}.

If the set of points $\{\b{x}^{(i)}\}_{i=1}^n$ is fixed, the cubature method in \eqref{eq:quadrature-generic} has $n$ degrees of freedom corresponding to the choice of the weights $\{w_i\}_{i=1}^n$.
The approach proposed by \citet{Sard1949} is a hybrid of the two classical approaches just described, calling for $m \leq n$ of these degrees of freedom to be used to ensure that $I_{\text{NI}}(f)$ is exact for $f$ in a given $m$-dimensional linear function space $\mathcal{F}$ and, if $m < n$, allocating the remaining $n-m$ degrees of freedom to select a minimal norm interpolant from a large class of functions $\mathcal{H}$. 
The approach of Sard is therefore exact for functions in the finite-dimensional set $\mathcal{F}$ and, at the same time, suitable for the integration of functions in the infinite-dimensional set $\mathcal{H}$.
Further background on Sard's method can be found in \cite{Larkin1974} and \cite{Karvonen2018}.

However, it is difficult to implement Sard's method, or indeed any of the classical approaches just discussed, in the Bayesian context, since
\begin{enumerate}
    \item the density $p$ can be evaluated pointwise only up to an intractable normalization constant;
    \item to construct weights one needs to evaluate the integrals of basis functions of $\mathcal{F}$ and of the interpolant $f_n$, which can be as difficult as evaluating the original integral.
\end{enumerate}
To circumvent these issues, in this paper we propose to combine Sard's approach to integration with Stein operators \citep{Stein1972,Gorham2015}, thus eliminating the need to access normalization constants and to exactly evaluate integrals.
A brief background on Stein operators is provided next.

\subsection{Stein Operators} \label{subsec: stein operators}

Let $\cdot$ denote the dot product $\b{a} \cdot \b{b}$ = $\b{a}^\top \b{b}$, $\nabla_{\b{x}}$ denote the gradient $\nabla_{\b{x}} = [\partial_{x_1},\dots,\partial_{x_d}]^\top$ and $\Delta_{\b{x}}$ denote the Laplacian $\Delta_{\b{x}} = \nabla_{\b{x}} \cdot \nabla_{\b{x}}$.
Let $\|\b{x}\| = (\b{x} \cdot \b{x})^{1/2}$ denote the Euclidean norm on~$\mathbb{R}^d$.
The construction that enables us to realize Sard's method in the Bayesian context is the Langevin Stein operator $\mathcal{L}$ \citep{Gorham2015} on $\mathbb{R}^d$, defined for sufficiently regular $g$ and $p$ as
\begin{align} \label{eq:stein-operator}
(\mathcal{L} g)(\b{x}) &= \Delta_{\b{x}} g(\b{x}) + \nabla_{\b{x}}  g(\b{x}) \cdot \nabla_{\b{x}}  \log{p(\b{x})}.
\end{align}
We refer to $\mathcal{L}$ as a Stein operator due to the use of equations of the form \eqref{eq:stein-operator} (up to a simple substitution) in the method of \citet{Stein1972} for assessing convergence in distribution and due to its property of producing functions whose integrals with respect to $p$ are zero under suitable conditions such as those described in Lemma~\ref{lemma:assum}.

\begin{lemma}\label{lemma:assum}
If $g \colon \mathbb{R}^d \rightarrow \mathbb{R}$ is twice continuously differentiable, $\log p : \mathbb{R}^d \rightarrow \mathbb{R}$ is continuously differentiable and $\| \nabla_{\b{x}} g(\b{x}) \| \leq C \| \b{x} \|^{-\delta} p(\b{x})^{-1}$ is satisfied for some $C \in \mathbb{R}$ and $\delta > d -1$, then
\begin{equation*}
\int (\mathcal{L} g)(\b{x}) p(\b{x}) \mathrm{d} \b{x}=0,
\end{equation*}
where $\mathcal{L}$ is the Stein operator in \eqref{eq:stein-operator}. 
\end{lemma}

\noindent
The proof is provided in Appendix \ref{app:ProofZero}.
Although our attention is limited to \eqref{eq:stein-operator}, the choice of Stein operator is not unique and other Stein operators can be derived using the generator method of \cite{Barbour1988} or using Schr\"odinger Hamiltonians \citep{Assaraf1999}. Contrary to the standard requirements for a Stein operator, the operator $\mathcal{L}$ in control functionals does not need to fully characterize convergence and, as a consequence, a broader class of functions $g$ can be considered than in more traditional applications of Stein's method \citep{Stein1972}. 

It follows that, if the conditions of Lemma \ref{lemma:assum} are satisfied by $g_n : \mathbb{R}^d \rightarrow \mathbb{R}$, the integral of a function of the form $f_n = c_n + \mathcal{L} g_n$ is simply $c_n$, the constant. 
The main challenge in developing control variates, or functionals, based on Stein operators is therefore to find a function $g_n$ such that the asymptotic variance $\sigma(f - f_n)^2$ is small. 
To explicitly minimize asymptotic variance, \cite{Mijatovic2018,Belomestny2019} and \cite{Brosse2019} restricted attention to particular Metropolis--Hastings or Langevin samplers for which asymptotic variance can be explicitly characterized.
The minimization of empirical variance has also been proposed and studied in the case where samples are independent \citep{Belomestny2017} and dependent \citep{Belomestny2019,Belomestny2019a}. 
For an approach that is not tied to a particular Markov kernel, authors such as \cite{Assaraf1999} and \cite{Mira2013} proposed to minimize mean squared error along the sample path, which corresponds to the case of an independent sampling method.
In a similar spirit, the constructions in \cite{Oates2017,Oates2019} and \cite{Barp2018} were based on a minimum-norm interpolant, where the choice of norm is decoupled from the mechanism from where the points are sampled.

In this paper we combine Sard's approach to integration with a minimum-norm interpolant construction in the spirit of \cite{Oates2017} and related work; this is described next.

\subsection{The Proposed Method}\label{ssec:proposedBSS}

In this section we first construct an infinite-dimensional space $\mathcal{H}$ and a finite-dimensional space $\mathcal{F}$ of functions; these will underpin the proposed semi-exact control functional method.

For the infinite-dimensional component, let $k \colon \mathbb{R}^d \times \mathbb{R}^d \to \mathbb{R}$ be a positive-definite \emph{kernel}, meaning that (i) $k$ is symmetric, with $k(\b{x},\b{y}) = k(\b{y},\b{x})$ for all $\b{x},\b{y} \in \mathbb{R}^d$, and (ii) the \emph{kernel matrix} $[\b{K}]_{i,j} = k(\b{x}^{(i)}, \b{x}^{(j)})$ is positive-definite for any distinct points $\{ \b{x}^{(i)} \}_{i=1}^n \subset \mathbb{R}^d$ and any $n \in \mathbb{N}$.
Recall that such a $k$ induces a unique \emph{reproducing kernel Hilbert space} $\mathcal{H}(k)$. 
This is a Hilbert space that consists of functions $g \colon \mathbb{R}^d \to \mathbb{R}$ and is equipped with an inner product $\langle \cdot , \cdot \rangle_{\mathcal{H}(k)}$. 
The kernel $k$ is such that $k(\cdot,\b{x}) \in \mathcal{H}(k)$ for all $\b{x} \in \mathbb{R}^d$ and it is \emph{reproducing} in the sense that $\langle g , k(\cdot, \b{x}) \rangle_{\mathcal{H}(k)} = g(\b{x})$ for any $g \in \mathcal{H}(k)$ and $\b{x} \in \mathbb{R}^d$.
For $\b{\alpha} \in \mathbb{N}_0^d$ the multi-index notation $\b{x}^{\b{\alpha}} := x_1^{\alpha_1} \cdots x_d^{\alpha_d}$ and $|\bm{\alpha}| = \alpha_1 + \dots + \alpha_d$ will be used.
If $k$ is twice continuously differentiable in the sense of \citet[][Definition~4.35]{Steinwart2008}, meaning that the derivatives 
\begin{equation*}
    \partial_{\b{x}}^{\b{\alpha}} \partial_{\b{y}}^{\b{\alpha}} k(\b{x}, \b{y}) = \frac{\partial^{2\abs[0]{\b{\alpha}}}}{\partial \b{x}^{\b{\alpha}} \partial \b{y}^{\b{\alpha}}} k(\b{x}, \b{y})
\end{equation*}
exist and are continuous for every multi-index $\b{\alpha} \in \mathbb{N}_0^d$ with $\abs[0]{ \b{\alpha}} \leq 2$, then 
\begin{equation} \label{eq:stein-kernel}
k_0(\b{x}, \b{y}) = \mathcal{L}_{\b{x}} \mathcal{L}_{\b{y}} k(\b{x}, \b{y}),
\end{equation}
where $\mathcal{L}_{\b{x}}$ stands for application of the Stein operator defined in~\eqref{eq:stein-operator} with respect to variable $\b{x}$, is a well-defined and positive-definite kernel~\citep[][Lemma 4.34]{Steinwart2008}. 
The kernel in \eqref{eq:stein-kernel} can be written as
\begin{equation} \label{eq:stein-kernel2}
\begin{split}
k_0(\b{x}, \b{y}) ={}& \Delta_{\b{x}} \Delta_{\b{y}} k(\b{x}, \b{y}) + \b{u}(\b{x})^\top \nabla_{\b{x}} \Delta_{\b{y}} k(\b{x},\b{y}) \\
&+ \b{u}(\b{y})^\top \nabla_{\b{y}} \Delta_{\b{x}} k(\b{x}, \b{y}) + \b{u}(\b{x})^\top \big[ \nabla_{\b{x}} \nabla_{\b{y}}^\top k(\b{x},\b{y}) \big] \b{u}(\b{y}),
\end{split}
\end{equation}
where $\nabla_{\b{x}} \nabla_{\b{y}}^\top k(\b{x}, \b{y})$ is the $d \times d$ matrix with entries $[\nabla_{\b{x}} \nabla_{\b{y}}^\top k(\b{x}, \b{y})]_{i,j} = \partial_{x_i} \partial_{y_j} k(\b{x}, \b{y})$ and $\b{u}(\b{x}) = \nabla_{\b{x}} \log p(\b{x})$. 
If $k$ is radial then \eqref{eq:stein-kernel2} can be simplified; see Appendix~\ref{appendix:kernels}.
Lemma~\ref{lem:boundary-kernel} establishes conditions under which the functions $\b{x} \mapsto k_0(\b{x},\b{y})$, $\b{y} \in \mathbb{R}^d$, and hence elements of the Hilbert space $\mathcal{H}(k_0)$ reproduced by $k_0$, have zero integral.
Let $\|\b{M}\|_{\text{OP}} = \sup_{\|\b{x}\| = 1} \|\b{M} \b{x}\|$ denote the operator norm of a matrix $\b{M} \in \mathbb{R}^{d \times d}$.

\begin{lemma} \label{lem:boundary-kernel} 
If $k \colon \mathbb{R}^d \times \mathbb{R}^d \rightarrow \mathbb{R}$ is twice continuously differentiable in each argument, $\log p : \mathbb{R}^d \rightarrow \mathbb{R}$ is continuously differentiable, $\| \nabla_{\b{x}}\nabla_{\b{y}}^\top k(\b{x},\b{y}) \|_{\textsc{OP}} \leq C(\b{y}) \| \b{x} \|^{-\delta} p(\b{x})^{-1}$ and $\| \nabla_{\b{x}}\Delta_{\b{y}}k(\b{x},\b{y}) \| \leq C(\b{y}) \| \b{x} \|^{-\delta} p(\b{x})^{-1}$ are satisfied for some $C: \mathbb{R}^d \rightarrow (0,\infty)$, and $\delta > d -1$, then 
\begin{equation} \label{eq: kernel ints to 0}
    \int k_0(\b{x}, \b{y}) p(\b{x}) \dif \b{x} = 0
\end{equation}
for every $\b{y} \in \mathbb{R}^d$, where $k_0$ is defined in \eqref{eq:stein-kernel}.
\end{lemma}

\noindent The proof is provided in Appendix \ref{app:ProofZeroKernel}.
The infinite-dimensional space $\mathcal{H}$ used in this work is exactly the reproducing kernel Hilbert space $\mathcal{H}(k_0)$.
The basic mathematical properties of $k_0$ and the Hilbert space it reproduces are contained in Appendix \ref{app: basic results on H} and these can be used to inform the selection of an appropriate kernel.

For the finite-dimensional component, let $\Phi$ be a linear space of twice-continuously differentiable functions with dimension $m-1$, $m \in \mathbb{N}$, and a basis $\{\phi_i\}_{i=1}^{m-1}$. 
Define then the space obtained by applying the differential operator~\eqref{eq:stein-operator} to $\Phi$ as $\mathcal{L} \Phi = \mathrm{span}\{ \mathcal{L} \phi_1, \ldots, \mathcal{L} \phi_{m-1} \}$.
If the pre-conditions of Lemma \ref{lemma:assum} are satisfied for each basis function $g = \phi_i$ then linearity of the Stein operator implies that $ \int (\mathcal{L}\phi) \mathrm{d}p = 0$ for every $\phi \in \Phi$.
Typically we will select $\Phi = \mathcal{P}^r$ as the polynomial space $\mathcal{P}^r = \mathrm{span} \{ \b{x}^{\b{\alpha}} : \, \b{\alpha} \in \mathbb{N}_0^d, \, 0 < \abs[0]{\b{\alpha}} \leq r \}$ for some non-negative integer $r$. 
Note that constant functions are excluded from $\mathcal{P}^r$ since they are in the null space of $\mathcal{L}$; when required we let $\mathcal{P}_0^r = \text{span}\{1\} \oplus \mathcal{P}^r$ denote the larger space with the constant functions included. 
The finite-dimensional space $\mathcal{F}$ is then taken to be $    \mathcal{F} = \mathrm{span} \{1\} \oplus \mathcal{L} \Phi = \mathrm{span} \{1, \mathcal{L} \phi_1, \ldots \mathcal{L} \phi_{m-1} \}$.

It is now possible to state the proposed method.
Following Sard, we approximate the integrand $f$ with a function $f_n$ that interpolates $f$ at the locations $\b{x}^{(i)}$, is exact on the $m$-dimensional linear space $\mathcal{F}$, and minimises a particular (semi-)norm subject to the first two constraints. 
It will occasionally be useful to emphasise the dependence of $f_n$ on $f$ using the notation $f_n(\cdot) = f_n(\cdot; f)$. 
The proposed interpolant takes the form
\begin{equation} \label{eq:interpolant}
    f_n(\b{x}) = b_1 + \sum_{i=1}^{m-1} b_{i+1} (\mathcal{L} \phi_i) (\b{x}) + \sum_{i=1}^n a_i k_0(\b{x}, \b{x}^{(i)}), 
\end{equation}
where the coefficients $\b{a} = (a_1,\ldots,a_n) \in \mathbb{R}^n$ and $\b{b} = (b_1,\ldots,b_m) \in \mathbb{R}^m$ are selected such that the following two conditions hold:
\begin{enumerate}
    \item $f_n(\b{x}^{(i)} ; f)  = f(\b{x}^{(i)})$ for $i = 1, \ldots, n$ (interpolation);
    \item $f_n(\cdot;f) = f(\cdot)$ whenever $f \in \mathcal{F}$ (semi-exactness).
\end{enumerate}
Since $\mathcal{F}$ is $m$-dimensional, these requirements correspond to the total of $n+m$ constraints.
Under weak conditions, discussed in Section \ref{subsec: computation}, the total number of degrees of freedom due to selection of $\b{a}$ and $\b{b}$ is equal to $n+m$ and the above constraints can be satisfied. 
Furthermore, the corresponding function $f_n$ can be shown to minimise a particular (semi-)norm on a larger space of functions, subject to the interpolation and exactness constraints \citep[to limit scope, we do not discuss this characterisation further but the semi-norm is defined in \eqref{eq: semi norm mt} and the reader can find full details in][Theorem 13.1]{Wendland2004}. 
Figure~\ref{fig:example-interpolation} illustrates one such interpolant.
The proposed estimator of the integral is then
\begin{equation} \label{eq:BSS-def}
    I_{\textsc{SECF}}(f) = \int f_n(\b{x}) p(\b{x}) \dif \b{x} ,
\end{equation}
a special case of \eqref{eqn:CV} (the interpolation condition causes the first term in \eqref{eqn:CV} to vanish) that we call a \emph{semi-exact control functional}.
The following is immediate from \eqref{eq:interpolant} and \eqref{eq:BSS-def}:

\begin{corollary} \label{cor: well defined}
Under the hypotheses of Lemma \ref{lemma:assum} for each $g = \phi_i$, $i = 1,\dots,m-1$, and Lemma \ref{lem:boundary-kernel}, it holds that, whenever the estimator $I_{\textsc{SECF}}(f)$ is well-defined, $I_{\textsc{SECF}}(f) = b_1$, where $b_1$ is the constant term in \eqref{eq:interpolant}.
\end{corollary}

\noindent The earlier work of \cite{Assaraf1999} and \cite{Mira2013} corresponds to $\b{a} = \b{0}$ and $\b{b} \neq \b{0}$, while setting $\b{b} = \b{0}$ in~\eqref{eq:interpolant} and ignoring the semi-exactness requirement recovers the unique minimum-norm interpolant in the Hilbert space $\mathcal{H}(k_0)$ where $k_0$ is reproducing, in the sense of~\eqref{eq:minimum-norm}.
The work of \cite{Oates2017} corresponds to $b_i = 0$ for $i = 2,\dots,m$.
It is therefore clear that the proposed approach is a strict generalization of existing work and can be seen as a compromise between semi-exactness and minimum-norm interpolation.

\begin{figure}[t]
  \centering
    \includegraphics[width=\textwidth]{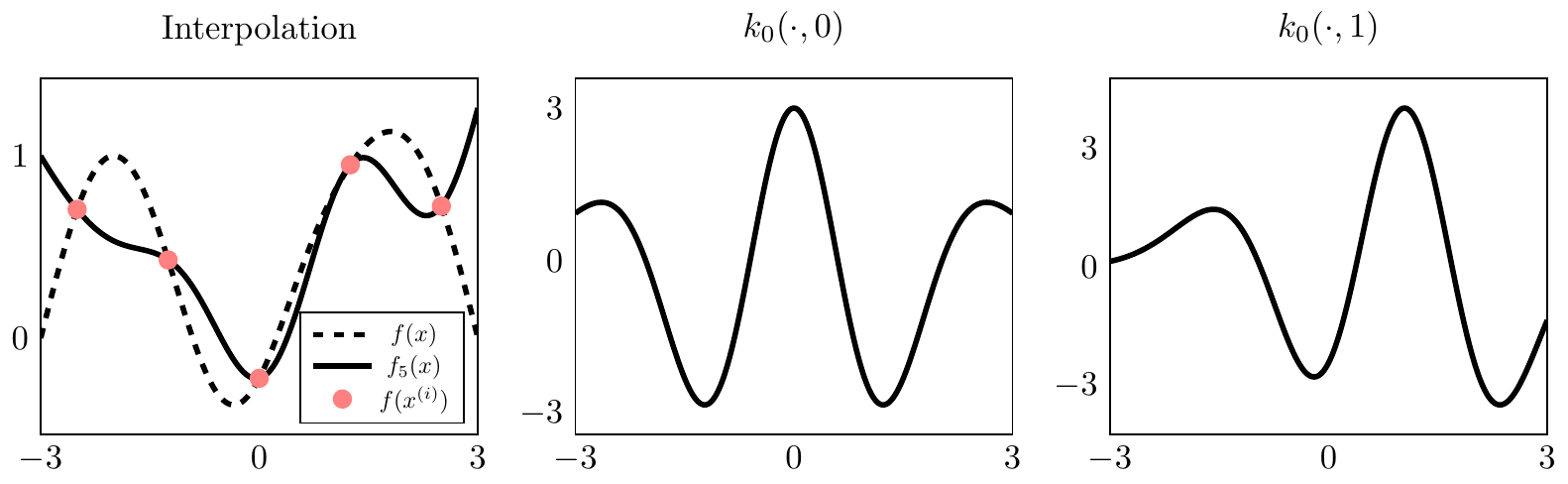}
    \caption{\emph{Left}: The interpolant $f_n$ from \eqref{eq:interpolant} at $n=5$ points to the function $f(x) = \sin(0.5\pi (x-1)) + \exp(-(x-0.5)^2)$ for the Gaussian density $p(x) = \mathcal{N}(x; 0 ,1)$. The interpolant uses the Gaussian kernel $k(x,y) = \exp(-(x-y)^2)$ and a polynomial parametric basis with $r=2$. \emph{Center} \& \emph{right}: Two translates $k_0(\cdot,y)$, $y \in \{0,1\}$, of the kernel~\eqref{eq:stein-kernel}.
    }
    \label{fig:example-interpolation}
\end{figure}

\subsection{Polynomial Exactness in the Bernstein-von-Mises Limit} \label{subsec: PE in BvM}

A central motivation for our approach is the prototypical case where $p$ is the density of a posterior distribution $P_{\b{x} \mid y_1,\dots,y_N}$ for a latent variable $\b{x}$ given independent and identically distributed data $y_1,\dots,y_N \sim P_{y_1,\dots,y_N \mid \b{x}}$.
Under regularity conditions discussed in Section 10.2 of  \cite{VanDerVaart1998}, the Bernstein-von-Mises theorem states that
\begin{equation*}
\Big\| P_{\b{x} \mid y_1,\dots,y_N} - \mathcal{N}\big(\hat{\b{x}}_N , N^{-1} I(\hat{\b{x}}_N)^{-1}\big) \Big\|_{\text{TV}} \rightarrow 0
\end{equation*}
where $\hat{\b{x}}_N$ is a maximum likelihood estimate for $\b{x}$, $I(\b{x})$ is the Fisher information matrix evaluated at $\b{x}$, $\|\cdot\|_{\text{TV}}$ is the total variation norm and convergence is in probability as $N \rightarrow \infty$ with respect to the law $P_{y_1,\dots,y_N \mid \b{x}}$ of the dataset.
In this limit, polynomial exactness of the proposed method can be established.
Indeed, for a Gaussian density $p$ with mean $\hat{\b{x}}_N \in \mathbb{R}^d$ and precision $N I(\hat{\b{x}}_N)$, if $\phi(\b{x}) = \b{x}^{\b{\alpha}}$ for a multi-index $\b{\alpha} \in \mathbb{N}_0^d$, then
\begin{equation*}
    (\mathcal{L} \phi)(\b{x}) =  \sum_{i=1}^d \alpha_i\left[ (\alpha_i-1) x_i^{\alpha_i-2} -\frac{N}{2} P_i(\b{x}) x_i^{\alpha_i-1} \right] \prod_{j \neq i} x_j^{\alpha_j},
\end{equation*}
where $P_i(\b{x}) = 2\b{e}_i^\top I(\hat{\b{x}}_N) (\b{x} - \hat{\b{x}}_N)$ and $\b{e}_i$ is the $i$th coordinate vector in $\mathbb{R}^d$. 
This allows us to obtain the following result, whose proof is provided in Appendix~\ref{app: proof of polyexact}:
\begin{lemma}\label{lem: polyexact}
Consider the Bernstein-von-Mises limit and suppose that the Fisher information matrix $I(\hat{\b{x}}_N)$ is non-singular.
Then, for the choice $\Phi = \mathcal{P}^r$, $r \in \mathbb{N}$, the estimator $I_\textsc{SECF}$ is exact on $\mathcal{F} = \mathcal{P}_0^r$.
\end{lemma}
Thus the proposed estimator is polynomially exact up to order $r$ in the Bernstein-von-Mises limit.
At finite $N$, when the limit has not been reached, the above argument can only be expected to approximately hold.

\subsection{Computation for the Proposed Method} \label{subsec: computation}

The purpose of this section is to discuss when the proposed estimator is well-defined and how it can be computed.
Define the $n \times m$ matrix
\begin{equation}
    \b{P} = \begin{bmatrix} 1 & \mathcal{L} \phi_1(\b{x}^{(1)}) & \cdots & \mathcal{L} \phi_{m-1}(\b{x}^{(1)}) \\ \vdots & \vdots & \ddots & \vdots \\ 1 & \mathcal{L} \phi_1 (\b{x}^{(n)}) & \cdots & \mathcal{L} \phi_{m-1}( \b{x}^{(n)}) \end{bmatrix}, \label{eq: def for P}
\end{equation}
which is sometimes called a \emph{Vandermonde} (or \emph{alternant}) matrix corresponding to the linear space $\mathcal{F}$.
Let $\b{K}_0$ be the $n \times n$ matrix with entries $[\b{K}_0]_{i,j} = k_0(\b{x}^{(i)}, \b{x}^{(j)})$ and let $\b{f}$ be the $n$-dimensional column vector with entries $[\b{f}]_i = f(\b{x}^{(i)})$. 

\begin{lemma} \label{lem: comput etc}
Let the $n \geq m$ points $\b{x}^{(i)}$ be distinct and $\mathcal{F}$-\emph{unisolvent}, meaning that the matrix $\b{P}$ in \eqref{eq: def for P} has full rank.
Let $k_0$ be a positive-definite kernel for which \eqref{eq: kernel ints to 0} is satisfied.
Then $I_{\textsc{SECF}}(f)$ is well-defined and the coefficients $\b{a}$ and $\b{b}$ are given by the solution of the linear system
\begin{equation} \label{eq:block-system}
    \begin{bmatrix} \b{K}_0 & \b{P} \\ \b{P}^\top & \b{0} \end{bmatrix} \begin{bmatrix} \b{a} \\ \b{b} \end{bmatrix} = \begin{bmatrix} \b{f} \\ \b{0} \end{bmatrix}.
\end{equation}
In particular,
\begin{equation}\label{eqn:BSS}
    I_{\textsc{SECF}}(f) = \b{e}_1^\top ( \b{P}^\top \b{K}_0^{-1} \b{P} )^{-1} \b{P}^\top \b{K}_0^{-1} \b{f}.
\end{equation}
\end{lemma}

\noindent The proof is provided in Appendix \ref{app:Computation proof}.
Notice that \eqref{eqn:BSS} is a linear combination of the values in~$\b{f}$ and therefore the proposed estimator is recognized as a cubature method of the form \eqref{eq:quadrature-generic} with weights 
\begin{equation} \label{eq: weights}
\b{w} = \b{K}_0^{-1} \b{P} ( \b{P}^\top \b{K}_0^{-1} \b{P} )^{-1} \b{e}_1.
\end{equation}

The requirement in Lemma \ref{lem: comput etc} for the $\b{x}^{(i)}$ to be distinct precludes, for example, the direct use of Metropolis--Hastings output. 
However, as emphasized in \cite{Oates2017} for control functionals and studied further in \cite{liu2017black,Hodgkinson2020}, the consistency of $I_{\text{SECF}}$ does \emph{not} require that the Markov chain is $p$-invariant.
It is therefore trivial to, for example, filter out duplicate states from Metropolis--Hastings output.

The solution of linear systems of equations defined by an $n \times n$ matrix $\b{K}_0$ and an $m \times m$ matrix $\b{P}^\top \b{K}_0^{-1} \b{P}$ entails a computational cost of $O(n^3 + m^3)$.
In some situations this cost may yet be smaller than the cost associated with evaluation of $f$ and $p$, but in general this computational requirement limits the applicability of the method just described.
In Appendix \ref{app: Nystrom} we therefore propose a computationally efficient approximation, $I_{\text{ASECF}}$, to the full method, based on a combination of the Nystr\"{o}m approximation \citep{williams2001using} and the well-known conjugate gradient method, inspired by the recent work of \cite{rudi2017falkon}.
All proposed methods are implemented in the \verb+R+ package \verb+ZVCV+ \citep{rZVCV}.

\section{Empirical Assessment}\label{sec: Empirical}

A detailed comparison of existing and proposed control variate and control functional techniques was performed.
Three examples were considered; Section \ref{sec: gaussian assessment} considers a Gaussian target, representing the Bernstein-von-Mises limit; Section \ref{subsec: capture} considers a setting where non-parametric control functional methods perform well; Section \ref{subsec: sonar} considers a setting where parametric control variate methods are known to be successful. 
In each case we determine whether or not the proposed semi-exact control functional method is competitive with the state-of-the-art.

Specifically, we compared the following estimators, which are all instances of $I_{\text{CV}}$ in \eqref{eqn:CV} for a particular choice of $f_n$, which may or may not be an interpolant:
\begin{itemize}
    \item Standard Monte Carlo integration, \eqref{eqn:MC}, based on Markov chain output.
    \item The control functional estimator recommended in \cite{Oates2017}, $I_{\text{CF}}(f) = (\b{1}^\top\b{K}_0^{-1}\b{1})^{-1} \b{1}^{\top} \b{K}_0^{-1}\b{f}$.
    \item The ``zero variance'' polynomial control variate method of \cite{Assaraf1999} and \cite{Mira2013}, $I_{\text{ZV}}(f) = \b{e}_1^\top (\b{P}^\top \b{P})^{-1}\b{P}^\top \b{f}$.
    \item The ``auto zero variance'' approach of \citet{South2019}, which uses 5-fold cross validation to automatically select (a) between the ordinary least squares solution $I_{\text{ZV}}$ and an $\ell_1$-penalised alternative (where the penalisation strength is itself selected using 10-fold cross-validation within the test dataset), and (b) the polynomial order.
    \item The proposed semi-exact control functional estimator, \eqref{eqn:BSS}.
    \item An approximation, $I_{\text{ASECF}}$, of \eqref{eqn:BSS} based on the Nystr\"{o}m approximation and the conjugate gradient method, described in Appendix \ref{app: Nystrom}.
\end{itemize}
Open-source software for implementing all of the above methods is available in the \texttt{R} package \texttt{ZVCV} \citep{rZVCV}. 
The same sets of $n$ samples were used for all estimators, in both the construction of $f_n$ and the evaluation of $I_{\text{CV}}$. 
For methods where there is a fixed polynomial basis we considered only orders $r=1$ and $r=2$, following the recommendation of \cite{Mira2013}.
For kernel-based methods, duplicate values of $\b{x}_i$ were removed (as discussed in Section~\ref{subsec: computation}) and Frobenius regularization was employed whenever the condition number of the kernel matrix $\b{K}_0$ was close to machine precision \citep{Higham1988}.
Several choices of kernel were considered, but for brevity in the main text we focus on the rational quadratic kernel $k(\b{x},\b{y}; \lambda) = (1+\lambda^{-2} \|\b{x}-\b{y}\|^2)^{-1}$. 
This kernel was found to provide the best performance across a range of experiments; a comparison to the Mat\'{e}rn and Gaussian kernels is provided in Appendix~\ref{app: effect of the kernel}.
The parameter $\lambda$ was selected using $5$-fold cross-validation, based again on performance across a spectrum of experiments; a comparison to the median heuristic \citep{Garreau2017} is presented in Appendix \ref{app: effect of the kernel}.

To ensure that our assessment is practically relevant, the estimators were compared on the basis of both statistical and computational efficiency relative to the standard Monte Carlo estimator. 
Statistical efficiency $\mathcal{E}(I_\text{CV})$ and computational efficiency $\mathcal{C}(I_\text{CV})$ of an estimator $I_\text{CV}$ of the integral $I$ are defined as
\begin{eqnarray*}
\mathcal{E}(I_\text{CV}) = \frac{\mathbb{E}\left[ (I_{\text{MC}} - I)^2 \right]}{\mathbb{E}\left[ (I_\text{CV} - I)^2 \right]}, \qquad
\mathcal{C}(I_\text{CV}) = \mathcal{E}(I_\text{CV}) \frac{T_{\text{MC}}}{T_\text{CV}}
\end{eqnarray*}
where $T_\text{CV}$ denotes the combined wall time for sampling the $\b{x}^{(i)}$ and computing the estimator $I_\text{CV}$.
For the results reported below, $\mathcal{E}$ and $\mathcal{C}$ were approximated using averages $\hat{\mathcal{E}}$ and $\hat{\mathcal{C}}$ over 100 realizations of the Markov chain output.

\subsection{Gaussian Illustration} \label{sec: gaussian assessment}

Here we consider a Gaussian integral that serves as an analytically tractable caricature of a posterior near to the Bernstein-von-Mises limit.
This enables us to assess the effect of the sample size $n$ and dimension $d$ on each estimator, in a setting that is not confounded by the idiosyncrasies of any particular MCMC method. 
Specifically, we set $p(\b{x}) = (2\pi)^{-d/2} \exp(-\|\b{x}\|^2 / 2)$ where $\b{x} \in \mathbb{R}^d$. 
For the parametric component we set $\Phi = \mathcal{P}^r$, so that (from Lemma \ref{lem: polyexact}) $I_{\text{SECF}}$ is exact on polynomials of order at most $r$; this holds also for $I_{\text{ZV}}$.
For the integrand $f : \mathbb{R}^d \rightarrow \mathbb{R}$, $d \geq 3$, we took
\begin{equation}\label{eqn:Gaussian_combin}
f(\b{x}) = 1 + x_2 + 0.1 x_1 x_2 x_3 + \sin(x_1) \exp[-(x_2 x_3)^2]
\end{equation}
in order that the integral is analytically tractable ($I(f) = 1$) and that no method will be exact. 

Figure \ref{fig:Gaussian} displays the statistical efficiency of each estimator for $10 \leq n \leq 1000$ and $3 \leq d \leq 100$. 
Computational efficiency is not shown since exact sampling from $p$ in this example is trivial. 
The proposed semi-exact control functional method performs consistently well compared to its competitors for this non-polynomial integrand. 
Unsurprisingly, the best improvements are for high $n$ and small $d$, where the proposed method results in a statistical efficiency over 100 times better than the baseline estimator and up to 5 times better than the next best method. 

\begin{figure}[t]
  \centering
    \includegraphics[width=0.8\textwidth]{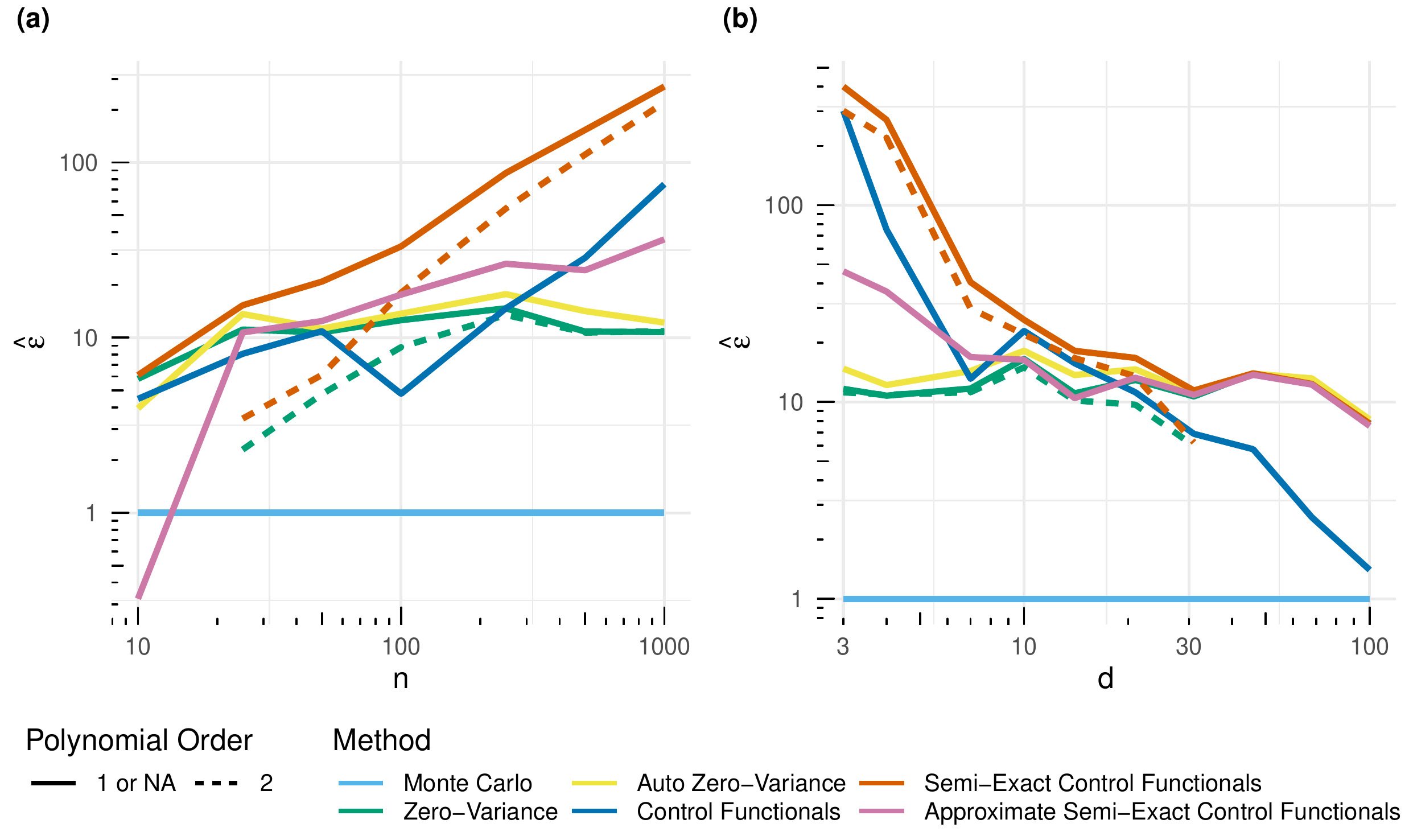}
    \caption{Gaussian example (a) estimated \textit{statistical efficiency} with $d=4$ and (b) estimated \textit{statistical efficiency} with $n=1000$ for integrand \eqref{eqn:Gaussian_combin}.
    }
    \label{fig:Gaussian}
\end{figure}

\subsection{Capture-Recapture Example} \label{subsec: capture}

The two remaining examples, here and in Section \ref{subsec: sonar}, are applications of Bayesian statistics described in \cite{South2019}.
In each case the aim is to estimate expectations with respect to a posterior distribution $P_{\b{x} \mid \b{y}}$ of the parameters $\b{x}$ of a statistical model based on $\b{y}$, an observed dataset. 
Samples $\b{x}^{(i)}$ were obtained using the Metropolis-adjusted Langevin algorithm \citep{Roberts1996}, which is a Metropolis-Hastings algorithm with proposal $\mathcal{N} ( \b{x}^{(i-1)} + h^2 \frac{1}{2}\b{\Sigma}\nabla_{\b{x}} \log P_{\b{x} \mid \b{y}}(\b{x}^{(i-1)} \mid \b{y}) , h^2 \b{\Sigma})$. 
Step sizes of $h=0.72$ for the capture-recapture example  and $h=0.3$ for the sonar example (see Section \ref{subsec: sonar}) were selected and an empirical approximation of the posterior covariance matrix was used as the pre-conditioner $\b{\Sigma} \in \mathbb{R}^{d\times d}$. 
Since the proposed method does not rely on the Markov chain being $P_{\b{x} \mid \b{y}}$-invariant we also repeated these experiments using the unadjusted Langevin algorithm \citep{Parisi1981,Ermak1975}, with similar results reported in Appendix~\ref{app: ULA results}.

In this first example, a Cormack--Jolly--Seber capture-recapture model \citep{Lebreton1992} is used to model data on the capture and recapture of the bird species \textit{Cinclus Cinclus} \citep{Marzolin1988}. The integrands of interest are the marginal posterior means $f_i(\b{x}) = x_i$ for $i=1,\ldots,11$, where $\b{x}=(\phi_1,\ldots,\phi_5,p_2,\ldots,p_6,\phi_6 p_7)$, $\phi_j$ is the probability of survival from year $j$ to $j+1$ and $p_j$ is the probability of being captured in year $j$. 
The likelihood is
\begin{align*}
\ell(\b{y}|\b{x}) \propto \prod_{i=1}^{6}\chi_i^{d_i} \prod_{k=i+1}^{7} \left[ \phi_i p_k \prod_{m=i+1}^{k-1} \phi_m (1-p_m) \right]^{y_{ik}},
\end{align*}
where $d_i=D_i-\sum_{k=i+1}^{7}y_{ik}$, $\chi_i=1-\sum_{k=i+1}^{7} \phi_i p_k \prod_{m=i+1}^{k-1} \phi_m (1-p_m)$ and the data $\b{y}$ consists of $D_i$, the number of birds released in year $i$, and $y_{ik}$, the number of animals caught in year $k$ out of the number released in year $i$, for $i=1,\ldots,6$ and $k=2,\ldots,7$. Following \citet{South2019}, parameters are transformed to the real line using $\tilde{x}_j=\log(x_j/(1-x_j))$ and the adjusted prior density for $\tilde{x}_j$ is $\exp(\tilde{x}_j)/(1+\exp(\tilde{x}_j))^2$, for $j=1,\ldots,11$.

\citet{South2019} found that non-parametric methods outperform standard parametric methods for this 11-dimensional example. 
The estimator $I_{\text{SECF}}$ combines elements of both approaches, so there is interest in determining how the method performs.
It is clear from Figure~\ref{fig:Recapture} that all variance reduction approaches are helpful in improving upon the vanilla Monte Carlo estimator in this example. The best improvement in terms of statistical and computational efficiency is offered by $I_{\text{SECF}}$, which also has similar performance to $I_{\text{CF}}$. 

\begin{figure}[t]
  \centering
    \includegraphics[width=0.8\textwidth]{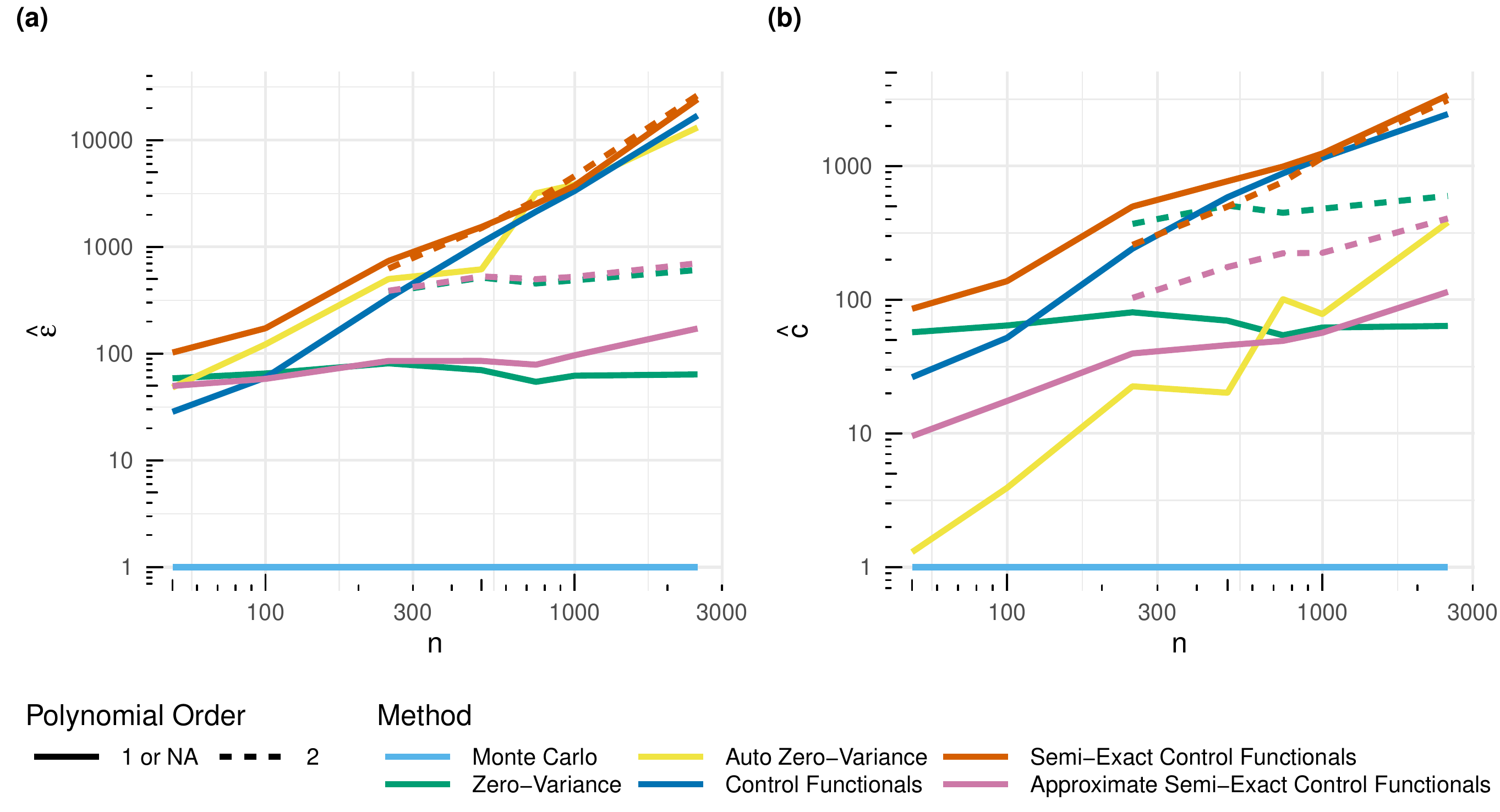}
    \caption{Capture-recapture example (a) estimated \textit{statistical efficiency} and (b) estimated \textit{computational efficiency}. Efficiency here is reported as an average over the 11 expectations of interest.
    }
    \label{fig:Recapture}
\end{figure}

\subsection{Sonar Example} \label{subsec: sonar}

Our final application is a 61-dimensional logistic regression example using data from \citet{Gorman1988} and \citet{Dheeru2017}. To use standard regression notation, the parameters are denoted $\b{\beta} \in \mathbb{R}^{61}$, the matrix of covariates in the logistic regression model is denoted $\b{X} \in \mathbb{R}^{208 \times 61}$ where the first column is all 1's to fit an intercept and the response is denoted $\b{y} \in \mathbb{R}^{208}$. In this application, $\b{X}$ contains information related to energy frequencies reflected from either a metal cylinder ($y=1$) or a rock ($y=0$). 
The log likelihood for this model is
\begin{equation*}\label{eqn:logistic}
\log \ell(\b{y},\b{X}|\b{\beta}) = \sum_{i=1}^{208} \left(y_i \b{X}_{i,\cdot}\b{\beta} - \log (1+\exp(\b{X}_{i,\cdot}\b{\beta}) )\right).
\end{equation*}
We use a $\mathcal{N}(0,5^2)$ prior for the predictors (after standardising to have standard deviation of 0.5) and $\mathcal{N}(0,20^2)$ prior for the intercept, following \citet{South2019,Chopin2017}, 
but we focus on estimating the more challenging integrand $f(\b{\beta}) = ( 1+\exp(-\tilde{\b{X}}\b{\beta}) )^{-1}$, which can be interpreted as the probability that observed covariates $\tilde{\b{X}}$ emanate from a metal cylinder. 
The gold standard of $I \approx 0.4971$ was obtained from a 10 million iteration Metropolis-Hastings \citep{Hastings1970} run with multivariate normal random walk proposal.

Figure \ref{fig:Sonar} illustrates the statistical and computational efficiency of estimators for various $n$ in this example. 
It is interesting to note that $I_{\text{SECF}}$ and $I_{\text{ASECF}}$ offer similar statistical efficiency to $I_{\text{ZV}}$, especially given the poor relative performance of $I_{\text{CF}}$. Since it is inexpensive to obtain the $m$ samples using the Metropolis-adjusted Langevin algorithm in this example, $I_{\text{ZV}}$ and $I_{\text{ASECF}}$ are the only approaches which offer improvements in computational efficiency over the baseline estimator for the majority of $n$ values considered, and even in these instances the improvements are marginal.

\begin{figure}[t]
  \centering
    \includegraphics[width=0.8\textwidth]{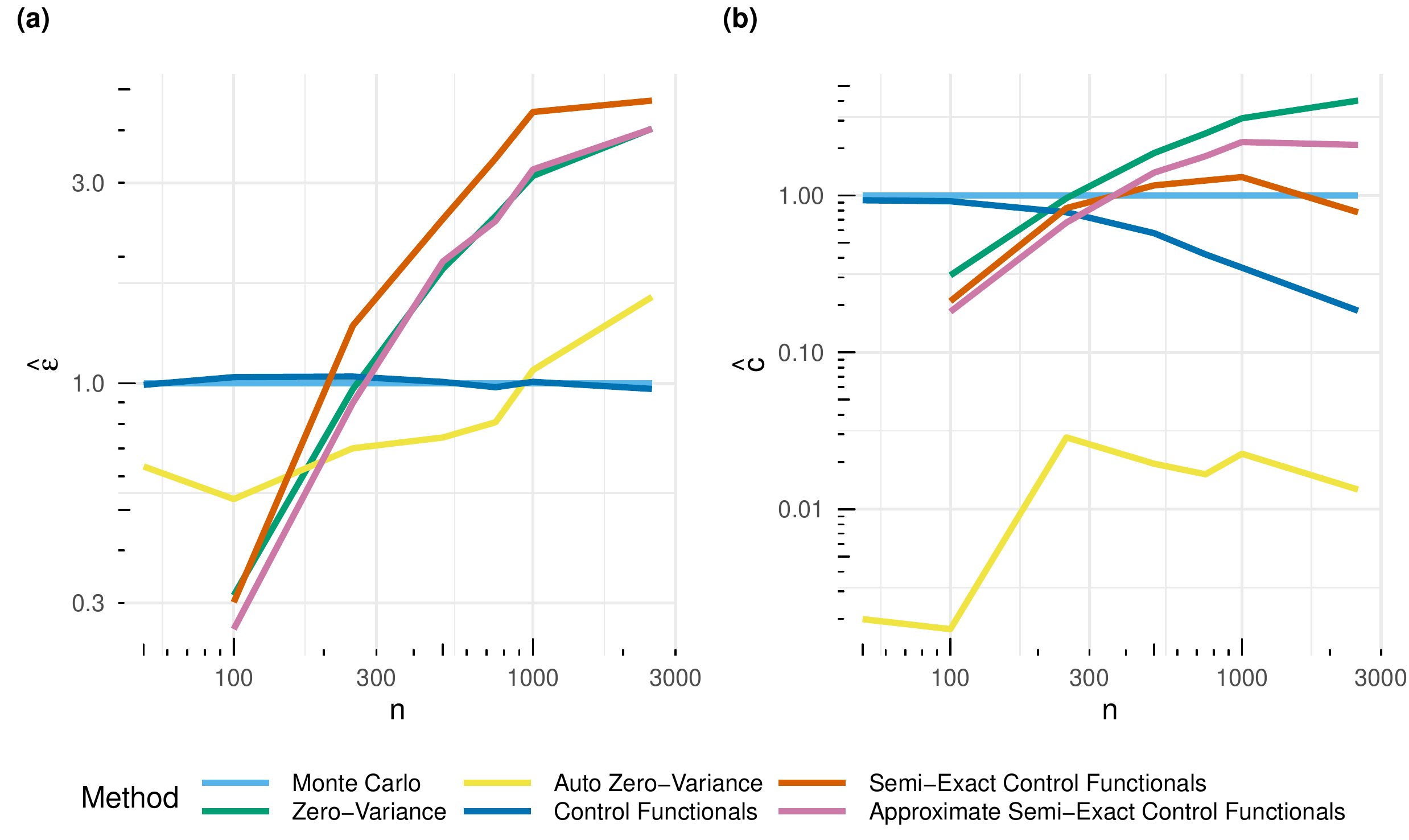}
    \caption{Sonar example (a) estimated \textit{statistical efficiency} and (b) estimated \textit{computational efficiency}.
    }
    \label{fig:Sonar}
\end{figure}

\section{Theoretical Properties and Convergence Assessment} \label{sec: theory}

In this section we provide a convergence result and discuss diagnostics that can be used to monitor the performance of the proposed method.
To this end, we introduce the semi-norm
\begin{align}
|f|_{k_0,\mathcal{F}} = \inf_{\substack{f = h + g \\ h \in \mathcal{F}, \, g \in \mathcal{H}(k_0) }} \|g\|_{\mathcal{H}(k_0)} , \label{eq: semi norm mt}
\end{align}
which is well-defined when the infimum is taken over a non-empty set, otherwise we define $|f|_{k_0,\mathcal{F}} = \infty$.

\subsection{Finite Sample Error and a Practical Diagnostic}

The following proposition provides a finite sample error bound:
\begin{proposition} \label{lem: KSD bound}
Let the hypotheses of Corollary \ref{cor: well defined} hold.
Then the integration error satisfies the bound
\begin{equation}
    \lvert I(f) - I_{\textsc{SECF}}(f) \rvert \leq |f|_{k_0,\mathcal{F}} \, (\b{w}^\top \b{K}_0 \b{w})^{1/2}  \label{eq: error bound KSD}
\end{equation}
where the weights $\b{w}$, defined in \eqref{eq: weights}, satisfy
\begin{equation*}
    \b{w} = \argmin_{ \b{v} \in \mathbb{R}^n} (\b{v}^\top \b{K}_0 \b{v})^{1/2} \quad \text{ s.t. } \quad \sum_{i=1}^n v_i h(\b{x}^{(i)}) = \int h(\b{x}) p(\b{x}) \dif \b{x} \quad \text{ for every } h \in \mathcal{F}.
\end{equation*}
\end{proposition}

The proof is provided in Appendix \ref{app: finite bound}.
The first quantity $|f|_{k_0,\mathcal{F}}$ in \eqref{eq: error bound KSD} can be approximated by $|f_n|_{k_0,\mathcal{F}}$ when $f_n$ is a reasonable approximation for $f$ and this can in turn can be bounded as $|f_n|_{k_0,\mathcal{F}} \leq (\b{a}^\top \b{K}_0 \b{a})^{1/2}$. 
The finiteness of $|f|_{k_0,\mathcal{F}}$ ensures the existence of a solution to the Stein equation, sufficient conditions for which are discussed in \cite{Mackey2016,Si2020}. 
The second quantity $(\b{w}^\top \b{K}_0 \b{w})^{1/2}$ in \eqref{eq: error bound KSD} is computable and is recognized as a \emph{kernel Stein discrepancy} between the empirical measure $\sum_{i=1}^n w_i \delta(\b{x}^{(i)})$ and the distribution whose density is $p$, based on the Stein operator $\mathcal{L}$ \citep{Chwialkowski2016,Liu2016}.
Note that our choice of Stein operator differs to that in \cite{Chwialkowski2016} and \citet{Liu2016}.
There has been substantial recent research into the use of kernel Stein discrepancies for assessing algorithm performance in the Bayesian computational context \citep{Gorham2017,Chen2018,Chen2019,singhal2019kernelized,Hodgkinson2020} and one can also exploit this discrepancy as a diagnostic for the performance of the semi-exact control functional.
The diagnostic that we propose to monitor is the product $(\b{w}^\top \b{K}_0 \b{w})^{1/2} (\b{a}^\top \b{K}_0 \b{a})^{1/2}$. This approach to error estimation was also suggested (outside the Bayesian context) in Section~5.1 of \cite{Fasshauer2011}.

Empirical results in Figure \ref{fig:KSD} suggest that this diagnostic provides a conservative approximation of the actual error. 
Further work is required to establish whether this diagnostic detects convergence and non-convergence in general.

\begin{figure}[!t]
\centering
\includegraphics[width=0.4\textwidth]{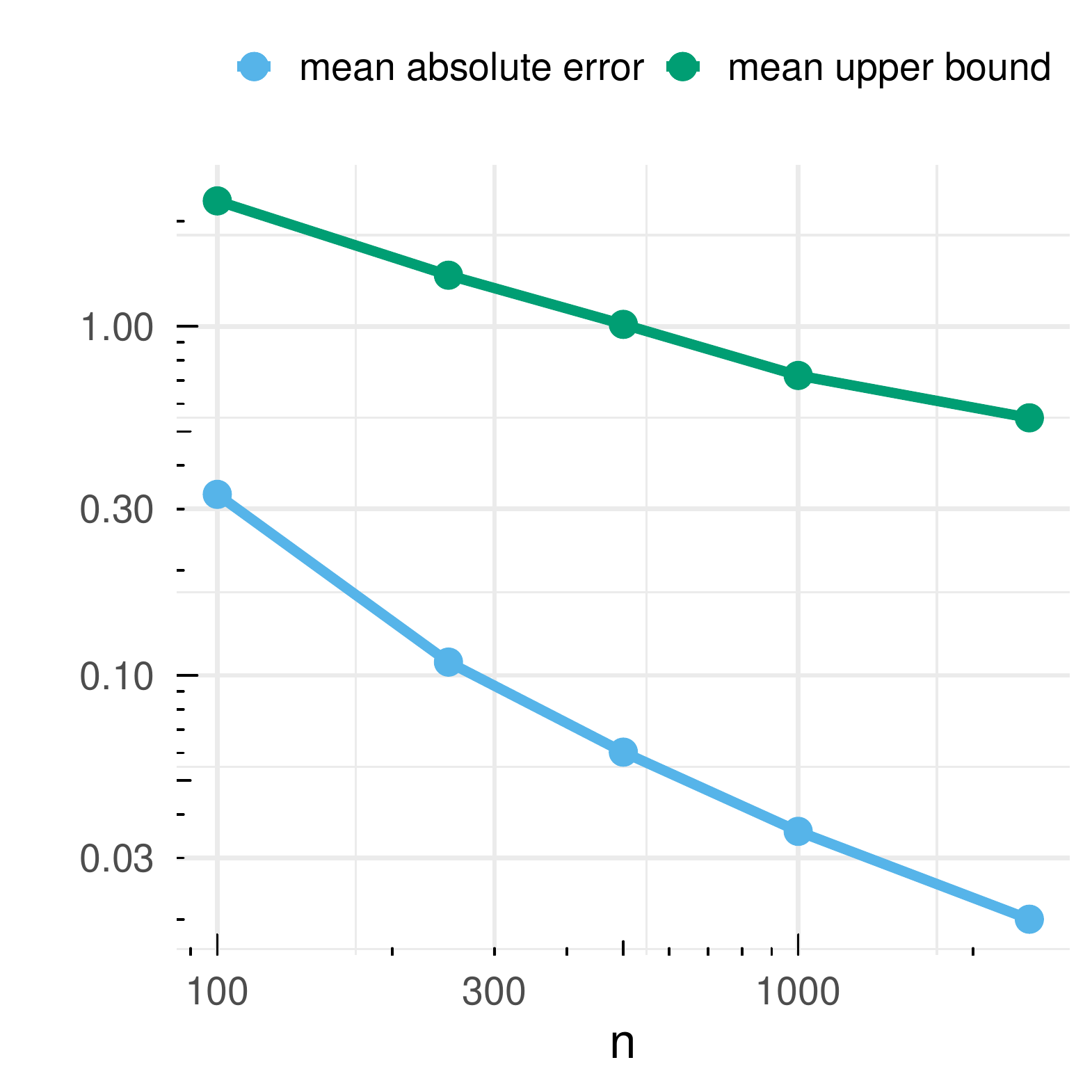}
\caption{The mean absolute error and mean of the approximate upper bound $(\b{w}^\top \b{K}_0 \b{w})^{1/2} (\b{a}^\top \b{K}_0 \b{a})^{1/2}$, for different values of $n$ in the sonar example of Section~\ref{subsec: sonar}. Both are based on the semi-exact control functional method with $\Phi = \mathcal{P}^1$.
}
\label{fig:KSD}
\end{figure}

\subsection{Consistency of the Estimator}

In this section we establish that, even in the biased sampling setting, the proposed estimator is consistent.
In what follows we consider an increasing number $n$ of samples $\bm{x}^{(i)}$ whilst the finite-dimensional space $\Phi$, with basis $\{\phi_1,\dots,\phi_{m-1}\}$, is held fixed. 
The samples $\bm{x}^{(i)}$ will be assumed to arise from a $V$-uniformly ergodic Markov chain; the reader is referred to Chapter 16 of \cite{meyn2012markov} for the relevant background.
Recall that the points $(\b{x}^{(i)})_{i=1}^n$ are called \emph{$\mathcal{F}$-unisolvent} if the matrix in \eqref{eq: def for P} has full rank.
It will be convenient to introduce an inner product $\langle \bm{u} , \bm{v} \rangle_n = \bm{u}^\top \bm{K}_0^{-1} \bm{v}$ and associated norm $\|\bm{u}\|_n = \langle \bm{u}, \bm{u} \rangle_n^{1/2}$.
Let $\Pi$ be the matrix that projects orthogonally onto the columns of $[\bm{\Psi}]_{i,j} := \mathcal{L} \phi_j(\bm{x}^{(i)})$ with respect to the $\langle \cdot , \cdot \rangle_n$ inner product. 

\begin{theorem} \label{thm: consistency}
Let the hypotheses of Corollary \ref{cor: well defined} hold and let $f$ be any function for which $|f|_{k_0,\mathcal{F}} < \infty$. 
Let $q$ be a probability density with $p/q > 0$ on $\mathbb{R}^d$ and consider a $q$-invariant Markov chain $(\bm{x}^{(i)})_{i=1}^n$, assumed to be $V$-uniformly ergodic for some $V : \mathbb{R}^d \rightarrow [1,\infty)$, such that 
\begin{enumerate}
    \item[A1.] $\sup_{\b{x} \in \mathbb{R}^d}  \; V(\b{x})^{-r} \; (p(\bm{x}) / q(\bm{x}))^4 \; k_0(\b{x},\b{x})^2   < \infty $ for some $0 < r < 1$;
    \item[A2.] the points $(\b{x}^{(i)})_{i=1}^n$ are almost surely distinct and $\mathcal{F}$-unisolvent;
    \item[A3.] $\lim\sup_{n \rightarrow \infty} \| \Pi \bm{1} \|_n / \| \bm{1} \|_n < 1$ almost surely.
\end{enumerate}
Then $| I_{\textsc{SECF}}(f) - I(f)| = O_P(n^{1/2})$.
\end{theorem}

The proof is provided in Appendix \ref{ap: consistency proof} and exploits a recent theoretical contribution in \cite{Hodgkinson2020}. 
Assumption A1 serves to ensure that $q$ is similar enough to $p$ that a $q$-invariant Markov chain will also explore the high probability regions of $p$, as discussed in \cite{Hodgkinson2020}.
Sufficient conditions for $V$-uniform ergodicity are necessarily Markov chain dependent. 
The case of the Metropolis-adjusted Langevin algorithm is discussed in \cite{Roberts1996,Chen2019} and, in particular, Theorem 9 of \cite{Chen2019} provides sufficient conditions for $V$-uniform ergodicity with $V(\bm{x}) = \exp(s \|\bm{x}\|)$ for all $s > 0$. 
Under these conditions, and with the rational quadratic kernel $k$ considered in Section \ref{sec: Empirical}, we have $k_0(\bm{x},\bm{x}) = O(\|\nabla_{\bm{x}} \log p(\bm{x})\|^2)$ and therefore A1 is satisfied whenever $(p(\bm{x}) / q(\bm{x}) ) \|\nabla_{\bm{x}} \log p(\bm{x})\| = O(\exp( t \|\bm{x}\|))$ for some $t > 0$; a weak requirement. 
Assumption A2 ensures that the finite sample size bound \eqref{eq: error bound KSD} is almost surely well-defined. 
Assumption A3 ensures the points in the sequence $(\bm{x}^{(i)})_{i=1}^n$ distinguish (asymptotically) the constant function from the functions $\{\phi_i\}_{i=1}^{m-1}$, which is a weak technical requirement.

\section{Discussion}\label{sec: Discussion}

The problem of approximating posterior expectations is well-studied and powerful control variate and control functional methods exist to improve the accuracy of Monte Carlo integration.
However, it is \textit{a priori} unclear which of these methods is most suitable for any given task.
This paper demonstrates how both parametric and non-parametric approaches can be combined into a single estimator that remains competitive with the state-of-the-art in all regimes we considered.
Moreover, we highlighted polynomial exactness in the Bernstein-von-Mises limit as a useful property that we believe can confer robustness of the estimator in a broad range of applied contexts.
The multitude of applications for these methods, and their availability in the \verb+ZVCV+ package \citep{rZVCV}, suggests they are well-placed to have a practical impact.

Several possible extensions of the proposed method can be considered.
For example, the parametric component $\Phi$ could be adapted to the particular $f$ and $p$ using a dimensionality reduction method.
Likewise, extending cross-validation to encompass the choice of kernel and even the choice of control variate or control functional estimator may be useful. The potential for alternatives to the Nystr\"{o}m approximation to further improve scalability of the method can also be explored. 
In terms of the points $\b{x}^{(i)}$ on which the estimator is defined, these could be optimally selected to minimize the error bound in \eqref{eq: error bound KSD}, for example following the approaches of \cite{Chen2018,Chen2019}.
Finally, we highlight a possible extension to the case where only stochastic gradient information is available, following \cite{Friel2016} in the parametric context.

\FloatBarrier

\paragraph{Acknowledgements}
CJO is grateful to Yvik Swan for discussion of Stein's method.
TK was supported by the Aalto ELEC Doctoral School and the Vilho, Yrj\"{o} and Kalle V\"{a}is\"{a}l\"{a} Foundation.
MG was supported by a Royal Academy of Engineering Research Chair and by the Engineering and Physical Sciences Research Council grants EP/T000414/1, EP/R018413/2, EP/P020720/2, EP/R034710/1, EP/R004889/1.
TK, MG and CJO were supported by the Lloyd's Register Foundation programme on data-centric engineering at the Alan Turing Institute, UK. CN and LFS were supported by the Engineering and Physical Sciences Research Council grants EP/S00159X/1 and EP/V022636/1. 
The authors are grateful for feedback from three anonymous Reviewers, an Associate Editor and an Editor, which led to an improved manuscript.

\bibliography{references}

\newpage
\appendix
\section*{Supplementary Material}

\section{Proof of Lemma \ref{lemma:assum}}\label{app:ProofZero}

Lemmas \ref{lemma:assum} and \ref{lem:boundary-kernel} are stylised versions of similar results that can be found in earlier work, such as \cite{Chwialkowski2016,Liu2016,Oates2017}.
Our presentation differs in that we provide a convenient explicit sufficient condition, on the tails of $\|\nabla g\|$ for Lemma \ref{lemma:assum}, and on the tails of $\|\nabla_{\bm{x}} \nabla_{\bm{y}}^\top k(\b{x},\b{y})\|$ and $\|\nabla_{\b{x}} \Delta_{\b{y}} k(\b{x},\b{y})\|$ for Lemma \ref{lem:boundary-kernel}, for their conclusions to hold.

\begin{proof}
The stated assumptions on the differentiability of $p$ and $g$ imply that the vector field $p(\b{x}) \nabla_{\b{x}} g(\b{x})$ is continuously differentiable on $\mathbb{R}^d$. 
The divergence theorem can therefore be applied, over any compact set $D \subset \mathbb{R}^d$ with piecewise smooth boundary $\partial D$, to reveal that
\begin{align*}
\int_{D} (\mathcal{L} g)(\b{x}) p(\b{x}) \mathrm{d} \b{x} &= \int_D [ \Delta_{\b{x}} g(\b{x}) + \nabla_{\b{x}} g(\b{x}) \cdot \nabla_{\b{x}} \log p(\b{x}) ] p(\b{x}) \mathrm{d}\b{x} \\
&= \int_{D} \left[ \frac{1}{p(\b{x})} \nabla_{\b{x}} \cdot (p(\b{x}) \nabla_{\b{x}} g(\b{x})) \right] p(\b{x}) \mathrm{d} \b{x} \\
&= \int_{D} \nabla_{\b{x}} \cdot (p(\b{x}) \nabla_{\b{x}} g(\b{x}))\mathrm{d} \b{x} \\
&=\oint_{\partial D} p(\b{x})\nabla_{\b{x}}  g(\b{x})\cdot \b{n}(\b{x}) \sigma(\mathrm{d}\b{x}),
\end{align*}
where $\b{n}(\b{x})$ is the unit normal vector at $\b{x} \in \partial D$ and $\sigma(\mathrm{d}\b{x})$ is the surface element at $\b{x} \in \partial D$. 
Next, we let $D = D_R = \{\b{x} : \|\b{x}\| \leq R\}$ be the ball in $\mathbb{R}^d$ with radius $R$, so that $\partial D_R$ is the sphere $S_R = \{\b{x} : \|\b{x}\| = R\}$.
The assumption $\|\nabla_{\b{x}} g(\b{x})\| \leq C \|\b{x}\|^{-\delta} p(\b{x})^{-1}$ in the statement of the lemma allows us to establish the bound
\begin{align*}
\left| \oint_{S_R} p(\b{x})\nabla_{\b{x}}  g(\b{x})\cdot \b{n}(\b{x}) \sigma(\mathrm{d}\b{x}) \right| \leq \oint_{S_R} \left| p(\b{x})\nabla_{\b{x}}  g(\b{x})\cdot \b{n}(\b{x}) \right| \sigma(\mathrm{d}\b{x}) &\leq \oint_{S_R} p(\b{x})\left\| \nabla_{\b{x}}  g(\b{x})\right\|  \sigma(\mathrm{d}\b{x})\\
&\leq \oint_{S_R} C\left\| \b{x}\right\|^{-\delta}  \sigma(\mathrm{d}\b{x}) \\
&= C R^{-\delta} \oint_{S_R} \sigma(\mathrm{d}\b{x}) \\
&= C R^{-\delta} \frac{2 \pi^{d/2}}{\Gamma(d/2)} R^{d-1},
\end{align*}
where in the first and second inequalities we used Jensen's inequality and Cauchy--Schwarz, respectively, and in the final equality we have made use of the surface area of $S_R$.
The assumption that $\delta > d - 1$ is then sufficient to obtain the result:
\begin{equation*}
\left| \int (\mathcal{L} g)(\b{x}) p(\b{x}) \mathrm{d}\b{x} \right| = \lim_{R \rightarrow \infty} \left| \oint_{S_R} p(\b{x})\nabla_{\b{x}}  g(\b{x})\cdot \b{n}(\b{x}) \sigma(\mathrm{d}\b{x}) \right| \; \leq \; \lim_{R \rightarrow \infty} C \frac{2 \pi^{d/2}}{\Gamma(d/2)} R^{d-1 - \delta} \; = \; 0.
\end{equation*}
This completes the argument.
\end{proof}

\section{Differentiating the Kernel} \label{appendix:kernels}

This appendix provides explicit forms of~\eqref{eq:stein-kernel2} for kernels $k$ that are radial. 
First we present a generic result in Lemma \ref{lem: k0 for radial} before specialising to the cases of the rational quadratic (Section \ref{subsec: RQK}), Gaussian (Section \ref{subsec: GK}) and Mat\'{e}rn (Section \ref{sec: matern}) kernels.

\begin{lemma} \label{lem: k0 for radial}
Consider a radial kernel $k$, meaning that $k$ has the form
\begin{equation*}
k(\b{x}, \b{y}) = \Psi(z), \qquad z = \| \b{x} - \b{y} \|^2,
\end{equation*}
where the function $\Psi \colon [0, \infty) \to \mathbb{R}$ is four times differentiable and $\b{x},\b{y} \in \mathbb{R}^d$.
Then \eqref{eq:stein-kernel2} simplifies to 
\begin{equation} \label{eq:stein-kernel-radial}
\begin{split}
k_0(\b{x},\b{y}) ={}& 16 z^2 \Psi^{(4)}(z) + 16 (2+d) z \Psi^{(3)}(z) + 4(2+d) d \Psi^{(2)}(z) \\
& + 4[2z \Psi^{(3)}(z) + (2+d) \Psi^{(2)}(z)] [\b{u}(\b{x}) - \b{u}(\b{y})]^\top (\b{x}-\b{y}) \\
& - 4 \Psi^{(2)}(z) \b{u}(\b{x})^\top (\b{x}-\b{y})(\b{x}-\b{y})^\top \b{u}(\b{y}) - 2 \Psi^{(1)}(z) \b{u}(\b{x})^\top \b{u}(\b{y}),
\end{split}
\end{equation}
where $\b{u}(\b{x}) = \nabla_{\b{x}} \log p(\b{x})$.
\end{lemma}
\begin{proof}
The proof is direct and based on have the following applications of the chain rule:
\begin{align*}
\nabla_{\b{x}} k(\b{x},\b{y}) & = 2 \Psi^{(1)}(z) (\b{x}-\b{y}), \\
\nabla_{\b{y}} k(\b{x},\b{y}) & = - 2 \Psi^{(1)}(z) (\b{x}-\b{y}), \\
\Delta_{\b{x}} k(\b{x},\b{y}) & = 4 z \Psi^{(2)}(z) + 2 d \Psi^{(1)}(z), \\
\Delta_{\b{y}} k(\b{x},\b{y}) & = 4 z \Psi^{(2)}(z) + 2 d \Psi^{(1)}(z), \\
\partial_{x_i} \partial_{y_j} k(\b{x}, \b{y}) &= -4\Psi^{(2)}(z) (x_i - y_i) (x_j - y_j) - 2\Psi^{(1)}(z) \delta_{ij}, \\
\nabla_{\b{x}} \Delta_{\b{y}} k(\b{x},\b{y}) & = 8 z \Psi^{(3)}(z) (\b{x}-\b{y}) + 4(2 + d) \Psi^{(2)}(z) (\b{x}-\b{y}), \\
\nabla_{\b{y}} \Delta_{\b{x}} k(\b{x},\b{y}) & = -8 z \Psi^{(3)}(z) (\b{x}-\b{y}) - 4(2 + d) \Psi^{(2)}(z) (\b{x}-\b{y}), \\
\Delta_{\b{x}} \Delta_{\b{y}} k(\b{x},\b{y}) & = 16z^2 \Psi^{(4)}(z) + 16(2+d) z \Psi^{(3)}(z) + 4(2+d) d \Psi^{(2)}(z).
\end{align*}
Upon insertion of these formulae into~\eqref{eq:stein-kernel2}, the desired result is obtained.
\end{proof}

Thus for kernels that are radial, it is sufficient to compute just the derivatives $\Psi^{(j)}$ of the radial part.

\subsection{Rational Quadratic Kernel} \label{subsec: RQK}

The rational quadratic kernel, 
\begin{equation*}
    \Psi(z) = (1+\lambda^{-2} z)^{-1},
\end{equation*}
has derivatives $\Psi^{(j)}(z) = (-1)^j \lambda^{-2j} j! (1+\lambda^{-2} z)^{-j - 1}$ for $j\geq1$.

\subsection{Gaussian Kernel} \label{subsec: GK}

For the Gaussian kernel we have $\Psi(z) = \exp(-z/\lambda^2)$. Consequently,
\begin{align*}
    \Psi^{(j)}(z) &= (-1)^j\lambda^{-2j} \exp(-z/\lambda^2),
\end{align*}
for $j\geq 1$.

\subsection{Mat\'ern Kernels} \label{sec: matern}

For a Mat\'ern kernel of smoothness $\nu > 0$ we have
\begin{equation*}
    \Psi(z) = b c^\nu z^{\nu/2} \, \mathrm{K}_\nu( c \sqrt{z}\,), \quad b = \frac{2^{1-\nu}}{\Gamma(\nu)}, \quad c = \frac{\sqrt{2\nu}}{\lambda},
\end{equation*}
where $\Gamma$ the Gamma function and $\mathrm{K}_\nu$ the modified Bessel function of the second kind of order $\nu$.
By the use of the formula $\partial_z K_\nu(z) = - \mathrm{K}_{\nu-1}(z) -\frac{\nu}{z} \mathrm{K}_\nu(z)$ we obtain
\begin{align*}
    \Psi^{(j)}(z) &= (-1)^j\frac{ b c^{\nu+j} }{2^j} z^{(\nu-j)/2} \, \mathrm{K}_{\nu-j}(c \sqrt{z}\,),
\end{align*}
for $j = 1,\ldots,4$.
In order to guarantee that the kernel is twice continously differentiable, so that $k_0$ in~\eqref{eq:stein-kernel} is well-defined, we require that $\lceil \nu \rceil > 2$. As a Mat\'ern kernel induces a reproducing kernel Hilbert space that is norm-equivalent to the standard Sobolev space of order $\nu + \frac{d}{2}$ \citep[][Example 5.7]{fasshauer2011reproducing}, the condition $\lceil \nu \rceil > 2$ implies, by the Sobolev imbedding theorem \citep[][Theorem 4.12]{adams2003sobolev}, that the functions in $\mathcal{H}(k)$ are twice continuously differentiable. Notice that $\Psi^{(3)}(z)$ and $\Psi^{(4)}(z)$ may not be defined at $z=0$, in which case the terms $16z^2 \Psi^{(4)}(z)$, $16(2+d)z\Psi^{(3)}(z)$ and $8z\Psi^{(3)}(z)$ in~\eqref{eq:stein-kernel-radial} must be interpreted as limits as $z \to 0$ from the right.

\section{Properties of $\mathcal{H}(k_0)$} \label{app: basic results on H}

The purpose of this appendix is to establish basic properties of the reproducing kernel Hilbert space $\mathcal{H}(k_0)$ of the kernel $k_0$ in~\eqref{eq:stein-kernel}.
For convenience, in this appendix we abbreviate $\mathcal{H}(k_0)$ to $\mathcal{H}$.
In Lemma~\ref{lem: linear transform rkhs} we clarify the reproducing kernel Hilbert space structure of~$\mathcal{H}$.
Then in Lemma~\ref{lem: sq integ} we establish square integrability of the elements of $\mathcal{H}$ and in Lemma~\ref{lem: contained in Sobolev} we establish the local smoothness of the elements of $\mathcal{H}$. 

To state these results we require several items of notation:
The notation $C^s(\mathbb{R}^d)$ denotes the set of $s$-times continuously differentiable functions on $\mathbb{R}^d$; i.e. $\partial^{\b{\alpha}} f \in C^0(\mathbb{R}^d)$ for all $|\b{\alpha}| \leq s$ where $C^0(\mathbb{R}^d)$ denotes the set of continuous functions on $\mathbb{R}^d$.
For two normed spaces $V$ and $W$, let $V \hookrightarrow W$ denote that $V$ is continuously embedded in $W$, meaning that $\|v\|_W \leq C \|v\|_V$ for all $v \in V$ and some constant $C \geq 0$.
In particular, we write $V \simeq W$ if and only if $V$ and $W$ are equal as sets and both $V \hookrightarrow W$ and $W \hookrightarrow V$.
Let $\mathcal{L}^2(p)$ denote the vector space of square integrable functions with respect to $p$ and equip this with the norm $\|h\|_{\mathcal{L}^2(p)} = ( \int h(\b{x})^2 p(\b{x}) \mathrm{d}\b{x} )^{1/2}$.
For $h : \mathbb{R}^d \rightarrow \mathbb{R}$ and $D \subset \mathbb{R}^d$ we let $h|_D : D \rightarrow \mathbb{R}$ denote the restriction of $h$ to $D$.

First we clarify the reproducing kernel Hilbert space structure of $\mathcal{H}$:

\begin{lemma}[Reproducing kernel Hilbert space structure of $\mathcal{H}$] \label{lem: linear transform rkhs}
Let $k : \mathbb{R}^d \times \mathbb{R}^d \rightarrow \mathbb{R}$ be a positive-definite kernel such that the regularity assumptions of Lemma \ref{lem:boundary-kernel}  are satisfied.
Let $\mathcal{H}$ denote the normed space of real-valued functions on $\mathbb{R}^d$ with norm
$$
\|h\|_{\mathcal{H}} = \inf_{\substack{h = \mathcal{L} g \\ g \in \mathcal{H}(k)}}  \|g\|_{\mathcal{H}(k)} .
$$
Then $\mathcal{H}$ admits the structure of a reproducing kernel Hilbert space with kernel $\kappa : \mathbb{R}^d \times \mathbb{R}^d \rightarrow \mathbb{R}$ given by $\kappa(\b{x},\b{y}) = k_0(\b{x},\b{y})$.
That is, $\mathcal{H} = \mathcal{H}(k_0)$.
Moreover, for $D \neq \emptyset$, let $\mathcal{H} |_D$ denote the normed space of real-valued functions on $D$ with norm
$$
\|h'\|_{\mathcal{H} |_D} = \inf_{\substack{h|_D = h' \\ h \in \mathcal{H}}} \|h\|_{\mathcal{H}} .
$$
Then $\mathcal{H} |_D$ is a reproducing kernel Hilbert space with kernel $\kappa |_D : D \times D \rightarrow \mathbb{R}$ given by $\kappa |_D(\b{x},\b{y}) = k_0(\b{x},\b{y})$.
That is, $\mathcal{H}|_D = \mathcal{H}(\kappa|_D)$.
\end{lemma}
\begin{proof}
The first statement is an immediate consequence of Theorem 5 in Section 4.1 of \cite{Berlinet2011}.
The second statement is an immediate consequence of Theorem 6 in Section 4.2 of \cite{Berlinet2011}.
\end{proof}

Next we establish when the elements of $\mathcal{H}$ are square-integrable functions with respect to $p$.

\begin{lemma}[Square integrability of $\mathcal{H}$] \label{lem: sq integ}
Let $\kappa$ be a radial kernel satisfying the pre-conditions of Lemma~\ref{lem: k0 for radial}.
If $u_i = \nabla_{x_i} \log p(\bm{x}) \in \mathcal{L}^2(p)$ for each $i = 1,\dots,d$, then $\mathcal{H} \hookrightarrow \mathcal{L}^2(p)$.
\end{lemma}
\begin{proof}
From the representer theorem and Cauchy--Schwarz we have
\begin{eqnarray}
\int h(\b{x})^2 p(\b{x}) \mathrm{d}\b{x} \; = \; \int \langle h , \kappa(\cdot,\b{x}) \rangle_{\mathcal{H}}^2 \, p(\b{x}) \mathrm{d}\b{x} \; \leq \; \|h\|_{\mathcal{H}}^2 \int \kappa(\b{x},\b{x}) p(\b{x}) \mathrm{d}\b{x} . \label{eqn: L2 bound}
\end{eqnarray}
Now, in the special case $k(\b{x},\b{y}) = \Psi(z)$, $z = \|\b{x}-\b{y}\|^2$, the conclusion of Lemma \ref{lem: k0 for radial} gives that $\kappa(\b{x},\b{x}) = 4 (2+d) d \Psi^{(2)}(0) - 2 \Psi^{(1)}(0) \|\b{u}(\b{x})\|^2$, from which it follows that 
\begin{eqnarray}
0 \leq \int \kappa(\b{x},\b{x}) p(\b{x}) \mathrm{d}\b{x} \; = \; 4 (2+d) d \Psi^{(2)}(0) - 2 \Psi^{(1)}(0) \int \|\b{u}(\b{x})\|^2 p(\b{x}) \mathrm{d}\b{x} \; = \; C^2 . \label{eq: kappa diag bound}
\end{eqnarray}
The combination of \eqref{eqn: L2 bound} and \eqref{eq: kappa diag bound} establishes that $\|h\|_{\mathcal{L}^2(p)} \leq C \|h\|_{\mathcal{H}}$, which is the claimed result.
\end{proof}

Finally we turn to the regularity of the elements of $\mathcal{H}$, as quantified by their smoothness over suitable bounded sets $D \subset \mathbb{R}^d$.
In what follows we will let $\mathcal{H}(k)$ be a reproducing kernel Hilbert space of functions in $\mathcal{L}^2(\mathbb{R}^d)$, the space of square Lebesgue integrable functions on $\mathbb{R}^d$, such that the norms
$$
\|h\|_{\mathcal{H}(k)} \simeq \|h\|_{W_2^r(\mathbb{R}^d)} = \Big( \textstyle \sum_{|\b{\alpha}| \leq r} \|\partial^{\b{\alpha}} h \|_{\mathcal{L}^2(\mathbb{R}^d)}^2 \Big)^{\frac{1}{2}}
$$
are equivalent.
The latter is recognized as the standard Sobolev norm; this space is denoted $W_2^r(\mathbb{R}^d)$.
For example, the Mat\'{e}rn kernel in Section~\ref{sec: matern} corresponds to $\mathcal{H}(k)$ with $r = \nu + \frac{d}{2}$.
The Sobolev embedding theorem implies that $W_2^r(\mathbb{R}^d) \subset C^0(\mathbb{R}^d)$ whenever $r > \frac{d}{2}$.

The following result establishes the smoothness of $\mathcal{H}$ in terms of the differentiability of its elements.
If the smoothness of $f$ is known then $k$ should be selected so that the smoothness of $\mathcal{H}$ matches it.

\begin{lemma}[Smoothness of $\mathcal{H}$] \label{lem: contained in Sobolev}
Let $r,s \in \mathbb{N}$ be such that $r > s + 2 + \frac{d}{2}$. If \sloppy{${\mathcal{H}(k) \simeq W_2^r(\mathbb{R}^d)}$} and $\log p \in C^{s+1}(\mathbb{R}^d)$, then, for any open and bounded set $D \subset \mathbb{R}^d$, we have $\mathcal{H}|_D \hookrightarrow W_2^s(D)$.
\end{lemma}
\begin{proof}
Under our assumptions, the kernel $\kappa|_D : D \times D \rightarrow \mathbb{R}$ from Lemma~\ref{lem: linear transform rkhs} is $s$-times continuously differentiable in the sense of Definition~4.35 of \cite{Steinwart2008}.
It follows from~Lemma~4.34 of \cite{Steinwart2008} that $\partial_{\b{x}}^{\b{\alpha}} \kappa|_D(\cdot,\b{x}) \in \mathcal{H}|_D$ for all $\b{x} \in D$ and $|\b{\alpha}| \leq s$.
From the reproducing property in $\mathcal{H}|_D$ and the Cauchy--Schwarz inequality we have, for $|\b{\alpha}| \leq s$,
\begin{align*}
|\partial^{\b{\alpha}} f(\b{x})| = \big| \langle f , \partial^{\b{\alpha}} \kappa|_D(\cdot, \b{x}) \rangle_{\mathcal{H}|_D} \big| \leq \|f\|_{\mathcal{H}|_D} \big\| \partial^{\b{\alpha}} \kappa|_D(\cdot, \b{x}) \big\|_{\mathcal{H}|_D} = \|f\|_{\mathcal{H}|_D} \left( \partial_{\b{x}}^{\b{\alpha}} \partial_{\b{y}}^{\b{\alpha}} \kappa|_D(\b{x},\b{y})|_{\b{y} = \b{x}} \right)^{1/2} .
\end{align*}
See also Corollary 4.36 of \cite{Steinwart2008}.
Thus it follows from the definition of $W_2^s(D)$ and the reproducing property that
\begin{align*}
\|f\|_{W_2^s(D)}^2 = \sum_{|\b{\alpha}| \leq s} \| \partial^{\b{\alpha}} f \|_{L_2(D)}^2 & \leq \| f \|_{\mathcal{H}|_D}^2 \sum_{|\b{\alpha}| \leq s} \big\| \b{x} \mapsto \partial_{\b{x}}^{\b{\alpha}} \partial_{\b{y}}^{\b{\alpha}} \kappa|_D(\b{x},\b{y})|_{\b{y} = \b{x}} \big\|_{L^2(D)}^2 \\
& = \| f \|_{\mathcal{H}|_D}^2 \big\| \b{x} \mapsto \kappa|_D(\b{x},\b{x}) \big\|_{W_2^s(D)}^2 .
\end{align*}
Now, from the definition of $\kappa$ and using the fact that $k$ is symmetric, we have
\begin{align*}
\kappa(\b{x},\b{x}) = \Delta_{\b{x}} \Delta_{\b{y}} k(\b{x},\b{y})|_{\b{y} = \b{x}} + 2 \b{u}(\b{x})^\top \nabla_{\b{x}} \Delta_{\b{y}} k(\b{x},\b{y})|_{\b{y} = \b{x}} + \b{u}(\b{x})^\top \big[ \nabla_{\b{x}} \nabla_{\b{y}}^\top k(\b{x},\b{y})|_{\b{y} = \b{x}} \big] \b{u}(\b{x}).
\end{align*}
Our assumption that $\mathcal{H}(k) \simeq W_2^r(\mathbb{R}^d)$ with $r > s + 2 + \frac{d}{2}$ implies that each of the functions $\b{x} \mapsto \Delta_{\b{x}} \Delta_{\b{y}} k(\b{x},\b{y})|_{\b{y} = \b{x}}$, $\nabla_{\b{x}} \Delta_{\b{y}} k(\b{x},\b{y})|_{\b{y} = \b{x}}$ and $\nabla_{\b{x}} \nabla_{\b{y}}^\top k(\b{x},\b{y})|_{\b{y} = \b{x}}$, are $C^s(\mathbb{R}^d)$.
In addition, our assumption that $\log p \in C^{s+1}(\mathbb{R}^d)$ implies that $\b{x} \mapsto \b{u}(\b{x}) \in C^s(\mathbb{R}^d)$.
Thus $\b{x} \mapsto \kappa(\b{x},\b{x})$ is $C^s(\mathbb{R}^d)$ and in particular the boundedness of $D$ implies that $\|\b{x} \mapsto \kappa|_D(\b{x},\b{x})\|_{W_2^s(D)} < \infty$ as required.
\end{proof}

\section{Proof of Lemma \ref{lem:boundary-kernel}}\label{app:ProofZeroKernel}

\begin{proof}

In what follows $C$ is a generic positive constant, independent of $\b{x}$ but possibly dependant on $\b{y}$, whose value can differ each time it is instantiated.
The aim of this proof is to apply Lemma \ref{lemma:assum} to the function $g(\b{x})=\mathcal{L}_{\b{y}} k(\b{x}, \b{y})$. 
Our task is to verify the pre-condition $\| \nabla_{\b{x}} g(\b{x}) \| \leq C \| \b{x} \|^{-\delta} p(\b{x})^{-1}$ for some \sloppy{${\delta > d-1}$}.
It will then follow from the conclusion of Lemma \ref{lemma:assum} that $\int k_0(\b{x}, \b{y}) p(\b{x}) \dif \b{x} = 0$ as required.
To this end, expanding the term $\| \nabla_{\b{x}} g(\b{x}) \|^2$, we have 
\begin{align}
    \| \nabla_{\b{x}} g(\b{x}) \|^2 ={}& \| \nabla_{\b{x}} \mathcal{L}_{\b{y}} k(\b{x}, \b{y}) \|^2  \nonumber  \\
     ={}& \big\| \nabla_{\b{x}} \Delta_{\b{y}} k(\b{x}, \b{y}) + \nabla_{\b{x}}[\nabla_{\b{y}}\log p(\b{y})\cdot\nabla_{\b{y}}k(\b{x},\b{y})] \big\|^2 \nonumber \\
     ={}& \|  \nabla_{\b{x}} \Delta_{\b{y}} k(\b{x}, \b{y}) \|^2 + 2\nabla_{\b{x}} \big[ \nabla_{\b{y}}\log p(\b{y})\cdot\nabla_{\b{y}}k(\b{x},\b{y}) \big]^\top \nabla_{\b{x}} \Delta_{\b{y}} k(\b{x}, \b{y}) \nonumber \\
     &+ \big\|   \nabla_{\b{x}} \big[ \nabla_{\b{y}}\log p(\b{y})\cdot\nabla_{\b{y}}k(\b{x},\b{y}) \big] \big\|^2 \nonumber \\
     ={}& \|  \nabla_{\b{x}} \Delta_{\b{y}} k(\b{x}, \b{y}) \|^2 + 2 \big\{ [\nabla_{\b{x}}\nabla_{\b{y}}^\top k(\b{x},\b{y})]^\top\nabla_{\b{y}}\log p(\b{y}) \big\}^\top \nabla_{\b{x}} \Delta_{\b{y}} k(\b{x}, \b{y})  \nonumber  \\
     &+ \big\|   [\nabla_{\b{x}}\nabla_{\b{y}}^\top k(\b{x},\b{y})]\nabla_{\b{y}}\log p(\b{y}) \big\|^2 \nonumber \\
     \begin{split}
     \leq{}& \|  \nabla_{\b{x}} \Delta_{\b{y}} k(\b{x}, \b{y}) \|^2 + 2 \big \|[\nabla_{\b{x}}\nabla_{\b{y}}^\top k(\b{x},\b{y})]^\top\nabla_{\b{y}}\log p(\b{y}) \big\| \| \nabla_{\b{x}} \Delta_{\b{y}} k(\b{x}, \b{y})\|  \\
     &+ \big\|   [\nabla_{\b{x}}\nabla_{\b{y}}^\top k(\b{x},\b{y})]\nabla_{\b{y}}\log p(\b{y}) \big\|^2 \label{eq: various bounds applied} 
     \end{split}
     \\
     \begin{split}
     \leq{}& \|  \nabla_{\b{x}} \Delta_{\b{y}} k(\b{x}, \b{y}) \|^2 + 2\|\nabla_{\b{x}}\nabla_{\b{y}}^\top k(\b{x},\b{y})\|_{\text{OP}}\|\nabla_{\b{y}}\log p(\b{y})\| \| \nabla_{\b{x}} \Delta_{\b{y}} k(\b{x}, \b{y})\| \\
     &+ \|   \nabla_{\b{x}}\nabla_{\b{y}}^\top k(\b{x},\b{y})\|_{\text{OP}}^2 \|\nabla_{\b{y}}\log p(\b{y})\|^2  \label{eq: various bounds applied 2} 
     \end{split}
     \\
     \begin{split}
     \leq{}& \big[ C \| \b{x} \|^{-\delta}p(\b{x})^{-1} \big]^2 + 2 \big[ C \| \b{x} \|^{-\delta}p(\b{x})^{-1} \big] \| \nabla_{\b{y}} \log p(\b{y}) \| \big[ C \| \b{x} \|^{-\delta}p(\b{x})^{-1} \big]  \\
     &+ \big[ C \| \b{x} \|^{-\delta}p(\b{x})^{-1} \big]^2 \|\nabla_{\b{y}} \log p(\b{y}) \|^2 \label{eqn: use bounds} 
     \end{split}
     \\
     \leq{}& C \| \b{x} \|^{-2\delta} p(\b{x})^{-2}  \nonumber 
\end{align}
as required.
Here \eqref{eq: various bounds applied} follows from the Cauchy--Schwarz inequality applied to the second term, \eqref{eq: various bounds applied 2} follows from the definition of the operator norm $\|\cdot\|_{\text{OP}}$ and \eqref{eqn: use bounds} employs the pre-conditions that we have assumed.
\end{proof}

\section{Proof of Lemma \ref{lem: polyexact}} \label{app: proof of polyexact}

\begin{proof}

Our first task is to establish that it is sufficient to prove the result in just the particular case $\hat{\b{x}}_N = \b{0}$ and $N^{-1} I(\hat{\b{x}}_N)^{-1} = \b{I}$, where $\b{I}$ is the $d$-dimensional identity matrix.
Indeed, if $\hat{\b{x}}_N\neq \b{0}$ or $N^{-1} I(\hat{\b{x}}_N)^{-1} \neq \b{I}$, then let $\b{t}(\b{x}) = \b{W}(\b{x}-\hat{\b{x}}_N)$ where $\b{W}$ is a non-singular matrix satisfying $\b{W}^\top \b{W} = N I(\hat{\b{x}}_N)$ so that \sloppy{${\b{t}(\b{x}) \sim \mathcal{N}(\b{0},\b{I})}$}. 
Under the same co-ordinate transformation the polynomial subspace 
$$
A=\mathcal{P}_0^r=\mathrm{span} \{ \b{x}^{\b{\alpha}} : \, \b{\alpha} \in \mathbb{N}_0^d, \, 0 \leq \abs[0]{\b{\alpha}} \leq r \}
$$
becomes $B=\mathrm{span} \{ \b{t}(\b{x})^{\b{\alpha}} : \, \b{\alpha} \in \mathbb{N}_0^d, \, 0 \leq \abs[0]{\b{\alpha}} \leq r \}$.
Exact integration of functions in $A$ with respect to $\mathcal{N}(\hat{\b{x}}_N , \b{I})$ corresponds to exact integration of functions in $B$ with respect to $\mathcal{N}(\b{0},\b{I})$.
Thus our first task is to establish that $B = A$.
Clearly $B$ is a linear subspace of $A$, since elements of $B$ can be expanded out into monomials and monomials generate $A$, so it remains to argue that $B$ is all of $A$. 
In what follows we will show that $\text{dim}(B) = \text{dim}(A)$ and this will complete the first part of the argument.

The co-ordinate transform $\b{t}$ is an invertible affine map on $\mathbb{R}^d$. 
The action of such a map $\b{t}$ on a set $S$ of functions on $\mathbb{R}^d$ can be defined as $\b{t}(S) = \{ \b{x} \rightarrow s(\b{t}(\b{x})) : s \in S \}$. Thus $B = \b{t}(A)$.
Let \sloppy{${\b{t}^*(\b{x}) = \b{W}^{-1}\b{x} + \hat{\b{x}}_n}$} and notice that this is also an invertible affine map on $\mathbb{R}^d$ with $\b{t}^*(\b{t}(\b{x})) = \b{x}$ being the identity map on $\mathbb{R}^d$. 
The composition of invertible affine maps on $\mathbb{R}^d$ is again an invertible affine map and thus $\b{t}^*\b{t}$ is also an invertible affine map on $\mathbb{R}^d$ and its action on a set is well-defined.
Considering the action of $\b{t}^*\b{t}$ on the set $A$ gives that $\b{t}^*(\b{t}(A)) = A$ and therefore $\b{t}(A)$ must have the same dimension as $A$. 
Thus $\text{dim}(A) = \text{dim}(t(A)) = \text{dim}(B)$ as claimed.

Our second task is to show that, in the case where $p$ is the density of $\mathcal{N}(\b{0},\b{I})$ and thus \sloppy{${\nabla_{\b{x}} \log p(\b{x}) = - \b{x}}$}, the set $\mathcal{F} = \text{span}\{1\} \oplus \mathcal{L}\mathcal{P}^r$ on which $I_{\text{CV}}$ is exact is equal to $\mathcal{P}_0^r$. 
Our proof proceeds by induction on the maximal degree $r$ of the polynomial.
For the base case we take $r = 1$:
\begin{align*}
    \text{span}\{1\} \oplus \mathcal{L}\mathcal{P}^1 &= \text{span}\{1\} \oplus \text{span} \big\{ \mathcal{L} x_j : j = 1,\ldots,d \big\} \\
    &= \text{span}\{1\} \oplus \text{span} \big\{ \Delta_{\b{x}} x_j + \nabla_{\b{x}} \log p(\b{x}) \cdot \nabla_{\b{x}}(x_j) : j = 1,\ldots,d \big\} \\
    &= \text{span}\{1\} \oplus \text{span} \big\{ 0-\b{x} \cdot \b{e}_j : j = 1,\ldots,d \big\} \\ 
    &= \text{span}\{1\} \oplus \text{span} \big\{ -x_j : j = 1,\ldots,d \big\} \\ 
    &= \mathcal{P}_0^1.
\end{align*}
For the inductive step we assume that $\text{span}\{1\} \oplus \mathcal{L}\mathcal{P}^{r-1}  = \mathcal{P}_0^{r-1}$ holds for a given $r \geq 2$ and aim to show that $\text{span}\{1\} \oplus \mathcal{L}\mathcal{P}^r = \mathcal{P}_0^r$.
Note that the action of $\mathcal{L}$ on a polynomial of order $r$ will return a polynomial of order at most $r$, so that $\text{span}\{1\} \oplus \mathcal{L}\mathcal{P}^r \subseteq \mathcal{P}_0^r$ and thus we need to show that $\mathcal{P}_0^r \subseteq \text{span}\{1\} \oplus \mathcal{L}\mathcal{P}^{r}$. 
Under the inductive assumption we have 
\begin{align*}
\text{span}\{1\} \oplus \mathcal{L}\mathcal{P}^{r} &= \text{span}\{1\} \oplus \left( \mathcal{L} \mathcal{P}^{r-1} \oplus \text{span}\big\{ \mathcal{L}\b{x}^{\b{\alpha}}: \, \b{\alpha} \in \mathbb{N}_0^d, \, \abs[0]{\b{\alpha}} = r \big\} \right) \nonumber \\
&= \left( \text{span}\{1\} \oplus \mathcal{L} \mathcal{P}^{r-1} \right) \oplus \text{span}\big\{ \mathcal{L}\b{x}^{\b{\alpha}}: \, \b{\alpha} \in \mathbb{N}_0^d, \, \abs[0]{\b{\alpha}} = r \big\}  \nonumber  \\
 & = \mathcal{P}_0^{r-1} \oplus \text{span}\big\{ \mathcal{L}\b{x}^{\b{\alpha}}: \, \b{\alpha} \in \mathbb{N}_0^d, \, \abs[0]{\b{\alpha}} = r \big\}  \nonumber \\
 & = \mathcal{P}_0^{r-1} \oplus \text{span} \big\{ \Delta_{\b{x}}\b{x}^{\b{\alpha}} + \nabla_{\b{x}}\b{x}^{\b{\alpha}} \cdot \nabla_{\b{x}}\log p(\b{x}): \, \b{\alpha} \in \mathbb{N}_0^d, \, \abs[0]{\b{\alpha}} = r \big\}  \nonumber \\
 & = \mathcal{P}_0^{r-1} \oplus \underbrace{ \text{span} \Bigg\{ \sum_{j=1}^d
    \alpha_j(\alpha_j-1)x_j^{\alpha_j-2}\prod_{k\neq j} x_k^{\alpha_k}
    - \sum_{j=1}^d\alpha_j\b{x}^{\b{\alpha}}: \, \b{\alpha} \in \mathbb{N}_0^d, \, \abs[0]{\b{\alpha}} = r \Bigg\}  }_{=: \mathcal{Q}^r}
\end{align*}
To complete the inductive step we must therefore show that, for each $\b{\alpha} \in \mathbb{N}_0^d$ with $|\b{\alpha}| = r$, we have $\b{x}^{\b{\alpha}} \in \text{span}\{1\} \oplus \mathcal{L} \mathcal{P}^r$.
Fix any $\b{\alpha} \in \mathbb{N}_0^d$ such that $\abs[0]{\b{\alpha}} = r$.
Then
$$
\phi(\b{x}) = \sum_{j=1}^d \alpha_j(\alpha_j-1)x_j^{\alpha_j-2}\prod_{k\neq j} x_k^{\alpha_k} - \sum_{j=1}^d\alpha_j \b{x}^{\b{\alpha}} \in \mathcal{Q}^r.
$$
and
$$
\varphi(\b{x}) =  \frac{1}{\b{1}^\top\b{\alpha}}\sum_{j=1}^d \alpha_j(\alpha_j-1)x_j^{\alpha_j-2}\prod_{k\neq j}x_k^{\alpha_k} \in \mathcal{P}_0^{r-1}
$$
because this polynomial is of order less than $r$.
Since $\varphi - (\b{1}^\top \b{\alpha})^{-1} \phi \in \mathcal{P}_0^{r-1} \oplus \mathcal{Q}^r = \text{span}\{1\} \oplus \mathcal{L} \mathcal{P}^r$ and
\begin{equation*}
    \varphi(\b{x}) - \frac{1}{\b{1}^\top \b{\alpha}} \phi(\b{x}) = \frac{\sum_{j=1}^d\alpha_j}{\b{1}^\top\b{\alpha}} \b{x}^{\b{\alpha}} = \b{x}^{\b{\alpha}},
\end{equation*}
we conclude that $\b{x}^{\b{\alpha}} \in \text{span}\{1\} \oplus \mathcal{L} \mathcal{P}^r$.
Thus we have shown that $\{\b{x}^{\b{\alpha}}  : \, \b{\alpha} \in \mathbb{N}_0^d, \, \abs[0]{\b{\alpha}} = r \} \subset \text{span}\{1\} \oplus \mathcal{L}\mathcal{P}^{r}$ and this completes the argument.
\end{proof}

\section{Proof of Lemma \ref{lem: comput etc}}
\label{app:Computation proof}

\begin{proof}

The assumptions that the $\b{x}^{(i)}$ are distinct and that $k_0$ is a positive-definite kernel imply that the matrix $\b{K}_0$ is positive-definite and thus non-singular.
Likewise, the assumption that the $\b{x}^{(i)}$ are $\mathcal{F}$-unisolvent implies that the matrix $\b{P}$ has full rank.
It follows that the block matrix in~\eqref{eq:block-system} is non-singular.
The interpolation and semi-exactness conditions in Section~\ref{ssec:proposedBSS} can be written in matrix form as
\begin{enumerate}
    \item $\b{K}_0 \, \b{a} + \b{P} \b{b} = \b{f}$ (interpolation);
    \item $\b{P}^\top \b{a} = \b{0}$ (semi-exact).
\end{enumerate}
The first of these is merely~\eqref{eq:interpolant} in matrix form. 
To see how $\b{P}^\top \b{a} = \b{0}$ is related to the semi-exactness requirement ($f_n = f$ whenever $f \in \mathcal{F}$), observe that for $f \in \mathcal{F}$ we have $\b{f} = \b{P} \b{c}$ for some $\b{c} \in \mathbb{R}^m$. Consequently, the interpolation condition should yield $\b{b} = \b{c}$ and $\b{a} = \b{0}$. The condition $\b{P}^\top \b{a} = \b{0}$ enforces that $\b{a} = \b{0}$ in this case: multiplication of the interpolation equation with $\b{a}^\top$ yields $\b{a}^\top \b{K}_0 \b{a} + \b{a}^\top \b{P} \b{b} = \b{a}^\top \b{P} \b{c}$, which is then equivalent to $\b{a}^\top \b{K}_0 \b{a} = \b{0}$. Because $\b{K}_0$ is positive-definite, the only possible $\b{a} \in \mathbb{R}^n$ is $\b{a} = \b{0}$ and $\b{P}$ having full rank implies that $\b{b} = \b{c}$.
Thus the coefficients $\b{a}$ and $\b{b}$ can be cast as the solution to the linear system
\begin{equation*}
    \begin{bmatrix} \b{K}_0 & \b{P} \\ \b{P}^\top & \b{0} \end{bmatrix} \begin{bmatrix} \b{a} \\ \b{b} \end{bmatrix} = \begin{bmatrix} \b{f} \\ \b{0} \end{bmatrix}.
\end{equation*}
From~\eqref{eq:block-system} we get
\begin{equation*}
    \b{b} = ( \b{P}^\top \b{K}_0^{-1} \b{P} )^{-1} \b{P}^\top \b{K}_0^{-1} \b{f},
\end{equation*}
where $\b{P}^\top \b{K}_0^{-1} \b{P}$ is non-singular because $\b{K}_0$ is non-singular and $\b{P}$ has full rank.
Recognising that \sloppy{${b_1 = \b{e}_1^\top \b{b}}$} for $\b{e}_1 = (1, 0, \ldots, 0) \in \mathbb{R}^m$ completes the argument.
\end{proof}

\section{Nystr\"{o}m Approximation and Conjugate Gradient}  \label{app: Nystrom}

In this appendix we describe how a Nystr\"{o}m approximation and the conjugate gradient method can be used to provide an approximation to the proposed method with reduced computational cost.
To this end we consider a function of the form
\begin{equation} 
    \tilde{f}_{n_0}(\b{x}) = \tilde{b}_1 + \sum_{i=1}^{m-1} \tilde{b}_{i+1} \mathcal{L} \phi_i (\b{x}) + \sum_{i=1}^{n_0} \tilde{a}_i k_0(\b{x}, \b{x}^{(i)}), \label{eq: Nystrom regression}
\end{equation}
where $n_0 \ll n$ represents a small subset of the $n$ points in the dataset. 
Strategies for selection of a suitable subset are numerous \citep[e.g.,][]{Alaoui2015,rudi2015less} but for simplicity in this work a uniform random subset was selected. 
Without loss of generality we denote this subset by the first $n_0$ indices in the dataset.
The coefficients $\b{a}$ and $\b{b}$ in the proposed method \eqref{eq:interpolant} can be characterized as the solution to a kernel least-squares problem, the details of which are reserved for Appendix \ref{app: kernel least squares}.
From this perspective it is natural to define the reduced coefficients $\tilde{\b{a}}$ and $\tilde{\b{b}}$ in \eqref{eq: Nystrom regression} also as the solution to a kernel least-squares problem, the details of which are reserved for Appendix \ref{app: Nystrom least squares}.
In taking this approach, the $(n+m)$-dimensional linear system in \eqref{eq:block-system} becomes the $(n_0+m)$-dimensional linear system
\begin{equation} \label{eq: Nystrom linear system}
\left[ \begin{array}{cc} \b{K}_{0,n_0,n}\b{K}_{0,n,n_0} + \b{P}_{n_0}\b{P}_{n_0}^\top & \b{K}_{0,n_0,n}\b{P} \\ \b{P}^\top\b{K}_{0,n,n_0} & \b{P}^\top\b{P}\end{array} \right] \left[ \begin{array}{c} \tilde{\b{a}} \\ \tilde{\b{b}} \end{array} \right] = \left[ \begin{array}{c} \b{K}_{0,n_0,n} \b{f}\\ \b{P}^\top \b{f} \end{array} \right].
\end{equation}
Here $\b{K}_{0,r,s}$ denotes the matrix formed by the first $r$ rows and the first $s$ columns of $\b{K}_0$.
Similarly $\b{P}_r$ denotes the first $r$ rows of $\b{P}$.
It can be verified that there is no approximation error when $n_0 = n$, with $\tilde{\b{a}} = \b{a}$ and $\tilde{\b{b}} = \b{b}$.
This is a simple instance of a Nystr\"{o}m approximation and it can be viewed as a random projection method \citep{smola2000sparse,williams2001using}.

The computational complexity of computing this approximation to the proposed method is 
$$
O(n n_0^2 + n m^2 + n_0^3 + m^3),
$$
which could still be quite high. 
For this reason, we now consider iterative, as opposed to direct, linear solvers for~\eqref{eq: Nystrom linear system}.
In particular, we employ the conjugate gradient method to approximately solve this linear system.
The performance of the conjugate gradient method is determined by the condition number of the linear system, and for this reason a preconditioner should be employed\footnote{A linear system $\b{A} \b{x} = \b{b}$ can be \emph{preconditioned} by an invertible matrix $\b{P}$ to produce $\b{P}^\top \b{A} \b{P} \b{z} = \b{P}^\top \b{b}$. The solution $\b{z}$ is related to $\b{x}$ via $\b{x} = \b{P} \b{z}$.}.
In this work we considered the preconditioner
\begin{equation*}
\left[ \begin{array}{cc} \b{B}_1 & \b{0} \\ \b{0} & \b{B}_2 \end{array} \right].
\end{equation*}
Following \cite{rudi2017falkon}, $\b{B}_1$ is the lower-triangular matrix resulting from a Cholesky decomposition
\begin{equation*}
\b{B}_1 \b{B}_1^\top = \left( \frac{n}{n_0} \b{K}_{0,n_0,n_0}^2 + \b{P}_{n_0} \b{P}_{n_0}^\top \right)^{-1} ,
\end{equation*}
the latter being an approximation to the inverse of $\b{K}_{0,n_0,n} \b{K}_{0,n,n_0} + \b{P}_{n_0} \b{P}_{n_0}^\top$ and obtained at \sloppy{${O(n_0^3 + m n_0^2)}$} cost. The matrix $\b{B}_2$ is
\begin{equation*}
\b{B}_2 \b{B}_2^\top = \big( \b{P}^\top\b{P}\big)^{-1} ,
\end{equation*}
which uses the pre-computed matrix $\b{P}^\top\b{P}$ and is of $O(m^3)$ complexity. 
Thus we obtain a preconditioned linear system
$$
\left[ \begin{array}{cc} \b{B}_1^\top ( \b{K}_{0,n_0,n} \b{K}_{0,n,n_0} + \b{P}_{n_0} \b{P}_{n_0}^\top ) \b{B}_1 & \b{B}_1^\top \b{K}_{0,n_0,n} \b{P}\b{B}_2 \\ \b{B}_2^\top \b{P}^\top \b{K}_{0,n,n_0} \b{B}_1 & \b{I} \end{array} \right] \left[ \begin{array}{c} \tilde{\tilde{\b{a}}} \\ \tilde{\tilde{\b{b}}} \end{array} \right] = \left[ \begin{array}{c} \b{B}_1^\top \b{K}_{0,n_0,n} \b{f} \\ \b{B}_2^\top \b{P}^\top \b{f} \end{array} \right] .
$$
The coefficients $\tilde{\b{a}}$ and $\tilde{\b{b}}$ of $\tilde{f}_{n_0}$ are related to the solution $(\tilde{\tilde{\b{a}}}, \tilde{\tilde{\b{b}}})$ of this preconditioner linear system via $\tilde{\tilde{\b{a}}} = \b{B}_1^{-1} \tilde{\b{a}}$ and $\tilde{\tilde{\b{b}}} =\b{B}_2^{-1}  \tilde{\b{b}}$, which is an upper-triangular linear system solved at quadratic cost.

The above procedure leads to a more computationally (time and space) efficient procedure, and we denote the resulting estimator as $I_{\text{ASECF}}(f) = \tilde{b}_1$.
Further extensions could be considered; for example non-uniform sampling for the random projection via leverage scores \citep{rudi2015less}.

For the examples in Section~\ref{sec: Empirical}, we consider $n_0 = \lceil\sqrt{n}\,\rceil$ where $\lceil \cdot \rceil$ denotes the ceiling function. We use the \verb+R+ package \verb+Rlinsolve+ to perform conjugate gradient, where we specify the tolerance to be~$10^{-5}$. 
The initial value for the conjugate gradient procedure was the choice of $\tilde{\tilde{\b{a}}}$ and $\tilde{\tilde{\b{b}}}$ that leads to the Monte Carlo estimate, $\tilde{\tilde{\b{a}}} = \b{0}$ and $\tilde{\tilde{\b{b}}} = \b{B}_2^{-1} \b{e}_{1}\frac{1}{n}\sum_{i=1}^n f(\b{x}^{(i)})$. In our examples, we did not see a computational speed up from the use of conjugate gradient, likely due to the relatively small values of $n$ involved.

\subsection{Kernel Least-Squares Characterization} \label{app: kernel least squares}

Here we explain how the interpolant $f_n$ in \eqref{eq:interpolant} can be characterized as the solution to the constrained kernel least-squares problem
\begin{equation*}
\argmin_{\b{a},\b{b}} \frac{1}{n}\sum_{i=1}^n \big[ f(\b{x}^{(i)}) - f_n(\b{x}^{(i)}) \big]^2 \quad \text{ s.t. } \quad f_n = f \quad \text{ for all } \quad f \in \mathcal{F} .
\end{equation*}
To see this, note that similar reasoning to that in Appendix \ref{app:Computation proof} allows us to formulate the problem using matrices as
\begin{eqnarray}
\argmin_{\b{a},\b{b}} \|\b{f} - \b{K}_0 \b{a} - \b{P} \b{b} \|^2 \quad \text{ s.t. } \quad \b{P}^\top \b{a} = \b{0} . \label{eq: least squares formulation of original}
\end{eqnarray}
This is a quadratic minimization problem subject to the constraint $\b{P}^\top \b{a} = \b{0}$ and therefore the solution is given by the Karush--Kuhn--Tucker matrix equation
\begin{eqnarray}
\left[ \begin{array}{ccc} \b{K}_0^2 & \b{K}_0 \b{P} & \b{P} \\ \b{P}^\top \b{K}_0 & \b{P}^\top \b{P} & \b{0} \\ \b{P}^\top & \b{0} & \b{0} \end{array} \right] \left[ \begin{array}{c} \b{a} \\ \b{b} \\ \b{c} \end{array} \right] = \left[ \begin{array}{c} \b{K} \b{f} \\ \b{P}^\top \b{f} \\ \b{0} \end{array} \right]  . \label{eq: KKT first}
\end{eqnarray}
Now, we are free to add a multiple, $\b{P}$, of the third row to the first row, which produces
\begin{eqnarray*}
\left[ \begin{array}{ccc} \b{K}_0^2 + \b{P} \b{P}^\top & \b{K}_0 \b{P} & \b{P} \\ \b{P}^\top \b{K}_0 & \b{P}^\top \b{P} & \b{0} \\ \b{P}^\top & \b{0} & \b{0} \end{array} \right] \left[ \begin{array}{c} \b{a} \\ \b{b} \\ \b{c} \end{array} \right] = \left[ \begin{array}{c} \b{K} \b{f} \\ \b{P}^\top \b{f} \\ \b{0} \end{array} \right] .
\end{eqnarray*}
Next, we make the \emph{ansatz} that $\b{c} = \b{0}$ and seek a solution to the reduced linear system
\begin{eqnarray*}
\left[ \begin{array}{cc} \b{K}_0^2 + \b{P} \b{P}^\top & \b{K}_0 \b{P} \\ \b{P}^\top \b{K}_0 & \b{P}^\top \b{P} \end{array} \right] \left[ \begin{array}{c} \b{a} \\ \b{b} \end{array} \right] = \left[ \begin{array}{c} \b{K} \b{f} \\ \b{P}^\top \b{f} \end{array} \right] .
\end{eqnarray*}
This is the same as
\begin{eqnarray*}
\left[ \begin{array}{ccc} \b{K}_0 & \b{P} \\ \b{P}^\top & \b{0} \end{array} \right] \left[ \begin{array}{ccc} \b{K}_0 & \b{P} \\ \b{P}^\top & \b{0} \end{array} \right] = \left[ \begin{array}{ccc} \b{K}_0 & \b{P} \\ \b{P}^\top & \b{0} \end{array} \right] \left[ \begin{array}{c} \b{f} \\ \b{0} \end{array} \right]
\end{eqnarray*}
and thus, if the block matrix can be inverted, we have 
\begin{eqnarray}
\left[ \begin{array}{ccc} \b{K}_0 & \b{P} \\ \b{P}^\top & \b{0} \end{array} \right] = \left[ \begin{array}{c} \b{f} \\ \b{0} \end{array} \right] \label{eq: KKT final}
\end{eqnarray}
as claimed.
Existence of a solution to \eqref{eq: KKT final} establishes a solution to the original system \eqref{eq: KKT first} and justifies the \emph{ansatz}.
Moreover, the fact that a solution to \eqref{eq: KKT final} exists was established in Lemma \ref{lem: comput etc}.

\subsection{Nystr\"{o}m Approximation} \label{app: Nystrom least squares}

To develop a Nystr\"{o}m approximation, our starting point is the kernel least-squares characterization of the proposed estimator in~\eqref{eq: least squares formulation of original}.
In particular, the same least-squares problem can be considered for the Nystr\"{o}m approximation in~\eqref{eq: Nystrom regression}:
\begin{eqnarray*}
\argmin_{\tilde{\b{a}},\tilde{\b{b}}} \|\b{f} - \b{K}_{0,n,n_0} \tilde{\b{a}} - \b{P} \tilde{\b{b}} \|_2^2 \quad \text{ s.t. } \quad \b{P}_{n_0}^\top \tilde{\b{a}} = \b{0} .
\end{eqnarray*}
This least-squares problem can be formulated as
\begin{align*}
&\argmin_{\tilde{\b{a}},\tilde{\b{b}}} (\b{f} - \b{K}_{0,n,n_0} \tilde{\b{a}} - \b{P} \tilde{\b{b}})^\top (\b{f} - \b{K}_{0,n,n_0} \tilde{\b{a}} - \b{P} \tilde{\b{b}}) \\
&=\argmin_{\tilde{\b{a}},\tilde{\b{b}}}
\left[ \b{f}^{\top}\b{f} 
- \b{f}^{\top}\b{K}_{0,n,n_0}\tilde{\b{a}}
- \b{f}^{\top}\b{P}\tilde{\b{b}}
- \tilde{\b{a}}^{\top}\b{K}_{0,n_0,n}\b{f}
+ \tilde{\b{a}}^{\top}\b{K}_{0,n_0,n}\b{K}_{0,n,n_0}\tilde{\b{a}} \right. \\
&\phantom{=\argmin_{\tilde{\b{a}},\tilde{\b{b}}}[ }\left.
+ \tilde{\b{a}}^{\top}\b{K}_{0,n_0,n}\b{P}\tilde{\b{b}}
- \tilde{\b{b}}^{\top}\b{P}^{\top}\b{f}
- \tilde{\b{b}}^{\top}\b{P}^{\top}\b{K}_{0,n,n_0}\tilde{\b{a}}
- \tilde{\b{b}}^{\top}\b{P}^{\top}\b{P}\tilde{\b{b}} \right]\\
&=\argmin_{\tilde{\b{a}},\tilde{\b{b}}}
\left[\begin{array}{c} \tilde{\b{a}} \\ \tilde{\b{b}} \end{array}\right]^{\top} 
\left[ \begin{array}{cc} \b{K}_{0,n_0,n}\b{K}_{0,n,n_0} & \b{K}_{0,n_0,n}\b{P} \\ \b{P}^\top\b{K}_{0,n,n_0} & \b{P}^\top\b{P} \end{array} \right]
\left[\begin{array}{c} \tilde{\b{a}} \\ \tilde{\b{b}} \end{array}\right]
- 2 \left[ \begin{array}{c} \b{K}_{0,n_0,n} \b{f} \\ \b{P}^\top \b{f} \end{array} \right]
\left[\begin{array}{c} \tilde{\b{a}} \\ \tilde{\b{b}} \end{array}\right]
+ \b{f}^{\top}\b{f}
\end{align*}
This is a quadratic minimization problem subject to the constraint $\b{P}_{n_0}^\top \tilde{\b{a}} = \b{0}$ and so the solution is given by the Karush--Kuhn--Tucker matrix equation
\begin{eqnarray}
\left[ \begin{array}{ccc} \b{K}_{0,n_0,n}\b{K}_{0,n,n_0} & \b{K}_{0,n_0,n}\b{P} & \b{P}_{n_0} \\ \b{P}^\top\b{K}_{0,n,n_0} & \b{P}^\top\b{P} & \b{0} \\ \b{P}_{n_0}^\top & \b{0} & \b{0} \end{array} \right] \left[ \begin{array}{c} \tilde{\b{a}} \\ \tilde{\b{b}} \\ \tilde{\b{c}} \end{array} \right] = \left[ \begin{array}{c} \b{K}_{0,n_0,n} \b{f}\\ \b{P}^\top \b{f} \\ \b{0} \end{array} \right] . \label{eq: KKT first Nystrom}
\end{eqnarray}
Following an identical argument to that in Appendix \ref{app: kernel least squares}, we first add $\b{P}_{n_0}$ times the third row to the first row to obtain
\begin{eqnarray*}
\left[ \begin{array}{ccc} \b{K}_{0,n_0,n}\b{K}_{0,n,n_0} + \b{P}_{n_0}\b{P}_{n_0}^\top & \b{K}_{0,n_0,n}\b{P} & \b{P}_{n_0} \\ \b{P}^\top\b{K}_{0,n,n_0} & \b{P}^\top\b{P} & \b{0} \\ \b{P}_{n_0}^\top & \b{0} & \b{0} \end{array} \right] \left[ \begin{array}{c} \tilde{\b{a}} \\ \tilde{\b{b}} \\ \tilde{\b{c}} \end{array} \right] = \left[ \begin{array}{c} \b{K}_{0,n_0,n} \b{f}\\ \b{P}^\top \b{f} \\ \b{0} \end{array} \right] .
\end{eqnarray*}
Taking again the \emph{ansatz} that $\tilde{\b{c}} = \b{0}$ requires us to solve the reduced linear system
\begin{eqnarray}
\left[ \begin{array}{cc} \b{K}_{0,n_0,n}\b{K}_{0,n,n_0} + \b{P}_{n_0}\b{P}_{n_0}^\top & \b{K}_{0,n_0,n}\b{P} \\ \b{P}^\top\b{K}_{0,n,n_0} & \b{P}^\top\b{P}\end{array} \right] \left[ \begin{array}{c} \tilde{\b{a}} \\ \tilde{\b{b}} \end{array} \right] = \left[ \begin{array}{c} \b{K}_{0,n_0,n} \b{f}\\ \b{P}^\top \b{f} \end{array} \right] . \label{eq: KTT Nystrom}
\end{eqnarray}
As in Appendix \ref{app: kernel least squares}, the existence of a solution to \eqref{eq: KTT Nystrom} implies a solution to \eqref{eq: KKT first Nystrom} and justifies the \emph{ansatz}.

\section{Sensitivity to the Choice of Kernel} \label{app: effect of the kernel}

In this appendix we investigate the sensitivity of kernel-based methods ($I_{\text{CF}}$, $I_{\text{SECF}}$ and $I_{\text{ASECF}}$) to the kernel and its parameter using the Gaussian example of Section \ref{sec: gaussian assessment}. Specifically we compare the three kernels described in Appendix \ref{appendix:kernels}, the Gaussian, Mat\'{e}rn and rational quadratic kernels, when the parameter,~$\lambda$, is chosen using either cross-validation or the median heuristic \citep{Garreau2017}. For the Mat\'{e}rn kernel, we fix the smoothness parameter at $\nu = 4.5$.

In the cross-validation approach,
\begin{eqnarray}
\lambda_\text{CV} \in \argmin \sum_{i=1}^5 \sum_{j=1}^{ n_5 } \big[ f( \b{x}^{(i,j)}) - f_{i,\lambda}( \b{x}^{(i,j)}) \big]^2, \label{eq: CV error}
\end{eqnarray}
where $n_5 := \lfloor n / 5 \rfloor$, $f_{i,\lambda}$ denotes an interpolant of the form \eqref{eq:interpolant} to $f$ at the points $\{\b{x}^{(i,j)} : j = 1,\dots, n_5 \big\}$ with kernel parameter $\lambda$, and $\b{x}^{(i,j)}$ is the $j$th point in the $i$th fold.  
In general~\eqref{eq: CV error} is an intractable optimization problem and we therefore perform a grid-based search. Here we consider $\lambda \in 10^{\{-1.5,-1,-0.5,0,0.5,1\}}$.

The median heuristic described in \citet{Garreau2017} is the choice of the bandwidth
\begin{equation*}
\tilde{\lambda} = \sqrt{ \frac{1}{2} \text{Med}\Big\{ \| \b{x}^{(i)} - \b{x}^{(j)} \|^2  \: : \: 1\leq i < j \leq n \Big\} }
\end{equation*}
for functions of the form $k(\b{x},\b{y}) = \varphi(\| \b{x} - \b{y} \|/\lambda)$, where $\text{Med}$ is the empirical median. This heuristic can be used for the Gaussian, Mat\'{e}rn and rational quadratic kernels, which all fit into this framework. 

Figures~\ref{fig:Gaussian_compare_N1000} and~\ref{fig:Gaussian_compare_d4} show the statistical efficiency of each combination of kernel and tuning approach for $n=1000$ and $d=4$, respectively. The outcome that the performance of $I_{\text{SECF}}$ and $I_{\text{ASECF}}$ are less sensitive to the kernel choice than $I_{\text{CF}}$ is intuitive when considering the fact that semi-exact control functionals enforce exactness on $f \in \mathcal{F}$.

\begin{figure}[!h]
\centering
\includegraphics[width=0.8\textwidth]{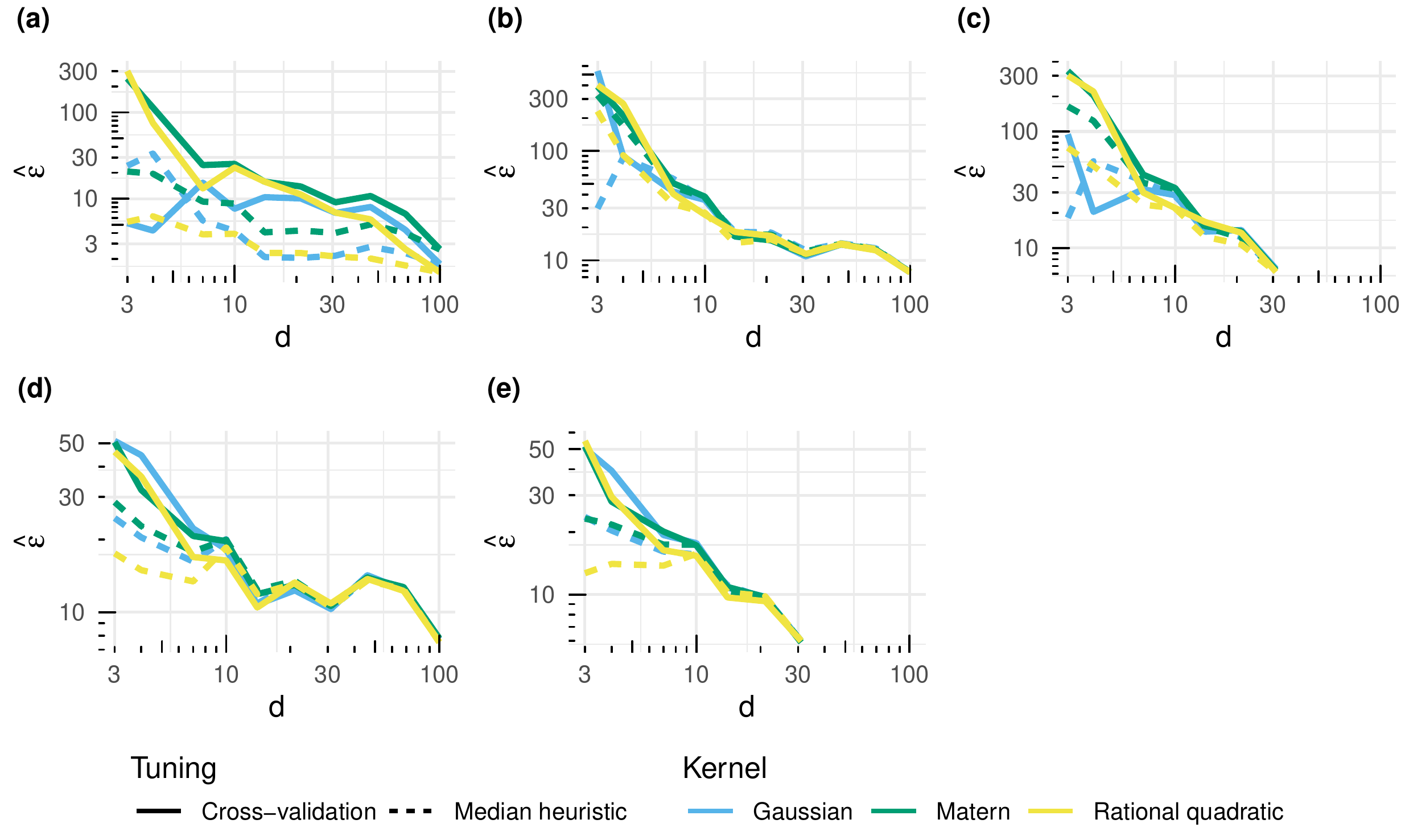}
\caption{Gaussian example, estimated statistical efficiency for $N=1000$ using different kernels and tuning approaches. The estimators are (a) $I_{\text{CF}}$, (b) $I_{\text{SECF}}$ with polynomial order $r=1$, (c) $I_{\text{SECF}}$ with $r=2$, (d) $I_{\text{ASECF}}$ with $r=1$ and (e) $I_{\text{ASECF}}$ with $r=2$.}
\label{fig:Gaussian_compare_N1000}
\end{figure}

\begin{figure}[!h]
\centering
\includegraphics[width=0.8\textwidth]{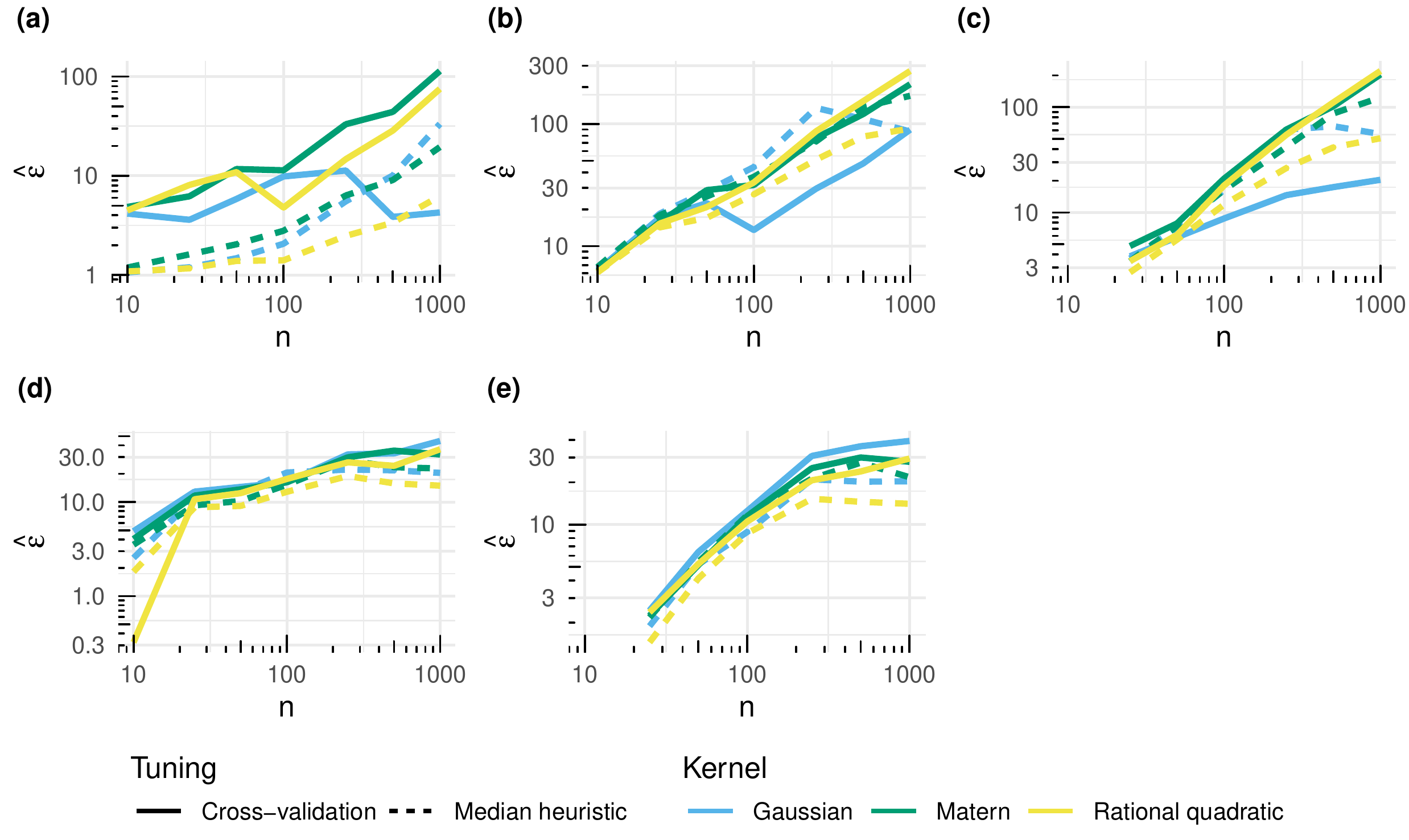}
\caption{Gaussian example, estimated statistical efficiency for $d=4$ using different kernels and tuning approaches. The estimators are (a) $I_{\text{CF}}$, (b) $I_{\text{SECF}}$ with polynomial order $r=1$, (c) $I_{\text{SECF}}$ with $r=2$, (d) $I_{\text{ASECF}}$ with $r=1$ and (e) $I_{\text{ASECF}}$ with $r=2$.
}
\label{fig:Gaussian_compare_d4}
\end{figure}

\FloatBarrier

\section{Results for the Unadjusted Langevin Algorithm}\label{app: ULA results}

Recall that the proposed method does not require that the $\b{x}^{(i)}$ form an empirical approximation to $p$.
It is therefore interesting to investigate the behaviour of the method when the $(\b{x}^{(i)})_{i=1}^\infty$ arise as a Markov chain that does not leave $p$ invariant.
Figures \ref{fig:Recapture_ULA} and \ref{fig:Sonar_ULA} show results when the unadjusted Langevin algorithm is used rather than the Metropolis-adjusted Langevin algorithm which is behind Figures \ref{fig:Recapture} and \ref{fig:Sonar} of the main text. The benefit of the proposed method for samplers that do not leave $p$ invariant is evident through its reduced bias compared to $I_{\text{ZV}}$ and $I_{\text{MC}}$ in Figure \ref{fig:Recapture_marginal1}.
Recall that the unadjusted Langevin algorithm \citep{Parisi1981,Ermak1975} is defined by
\begin{equation*}
\b{x}^{(i+1)} = \b{x}^{(i)} + \frac{h^2}{2}\b{\Sigma}\nabla_{\b{x}} \log P_{\b{x} \mid \b{y}}(\b{x}^{(i)} \mid \b{y}) + \epsilon_{i+1},
\end{equation*}
for $i=1,\ldots,n-1$ where $\b{x}^{(1)}$ is a fixed point with high posterior support and $\epsilon_{i+1}\sim \mathcal{N}(\b{0},h^2\b{\Sigma})$. Step sizes of $h=0.9$ for the sonar example and $h=1.1$ for the capture-recapture example were selected.

\begin{figure}[h!]
  \centering
\includegraphics[width=0.8\textwidth]{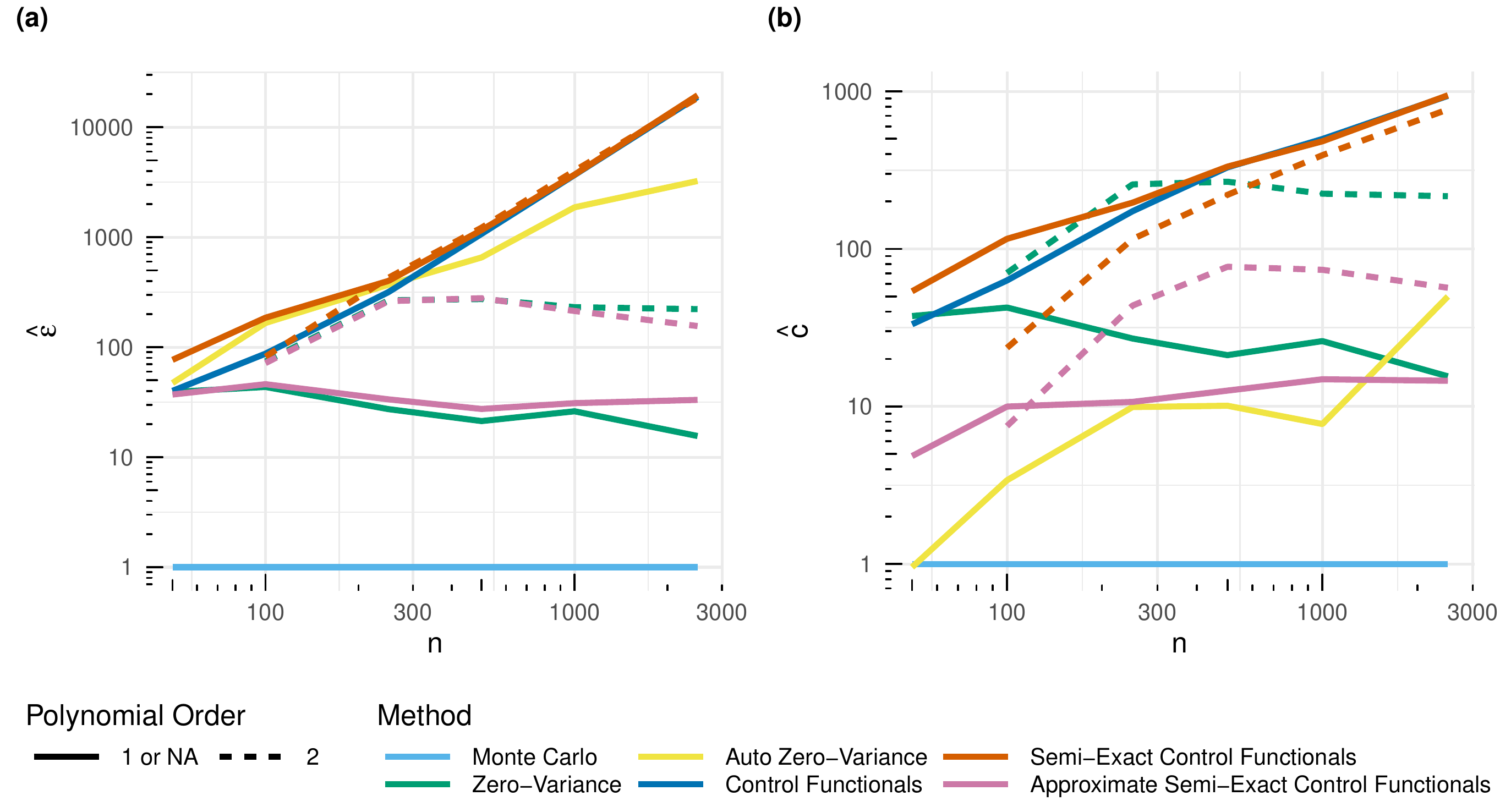}
 \caption{Recapture example (a) estimated \textit{statistical efficiency} and (b) estimated \textit{computational efficiency} when the unadjusted Langevin algorithm is used in place of the Metropolis-adjusted Langevin algorithm. Efficiency here is reported as an average over the 11 expectations of interest.
 }
    \label{fig:Recapture_ULA}
\end{figure}

\begin{figure}
  \centering
\includegraphics[width=0.8\textwidth]{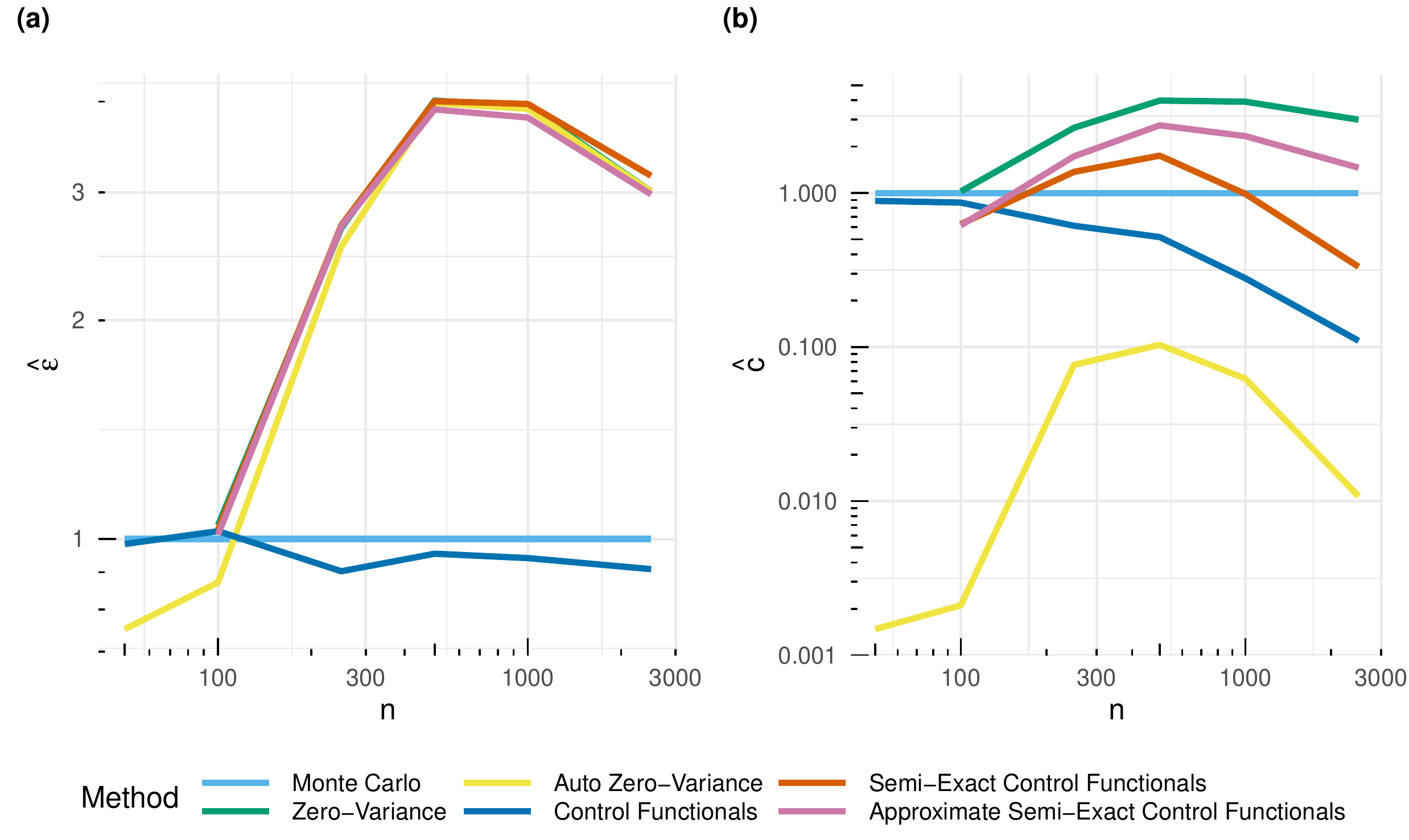}
\caption{Sonar example (a) estimated \textit{statistical efficiency} and (b) estimated \textit{computational efficiency} when the unadjusted Langevin algorithm is used in place of the Metropolis-adjusted Langevin algorithm.
}
\label{fig:Sonar_ULA}
\end{figure}

\begin{figure}
  \centering
\includegraphics[width=0.8\textwidth]{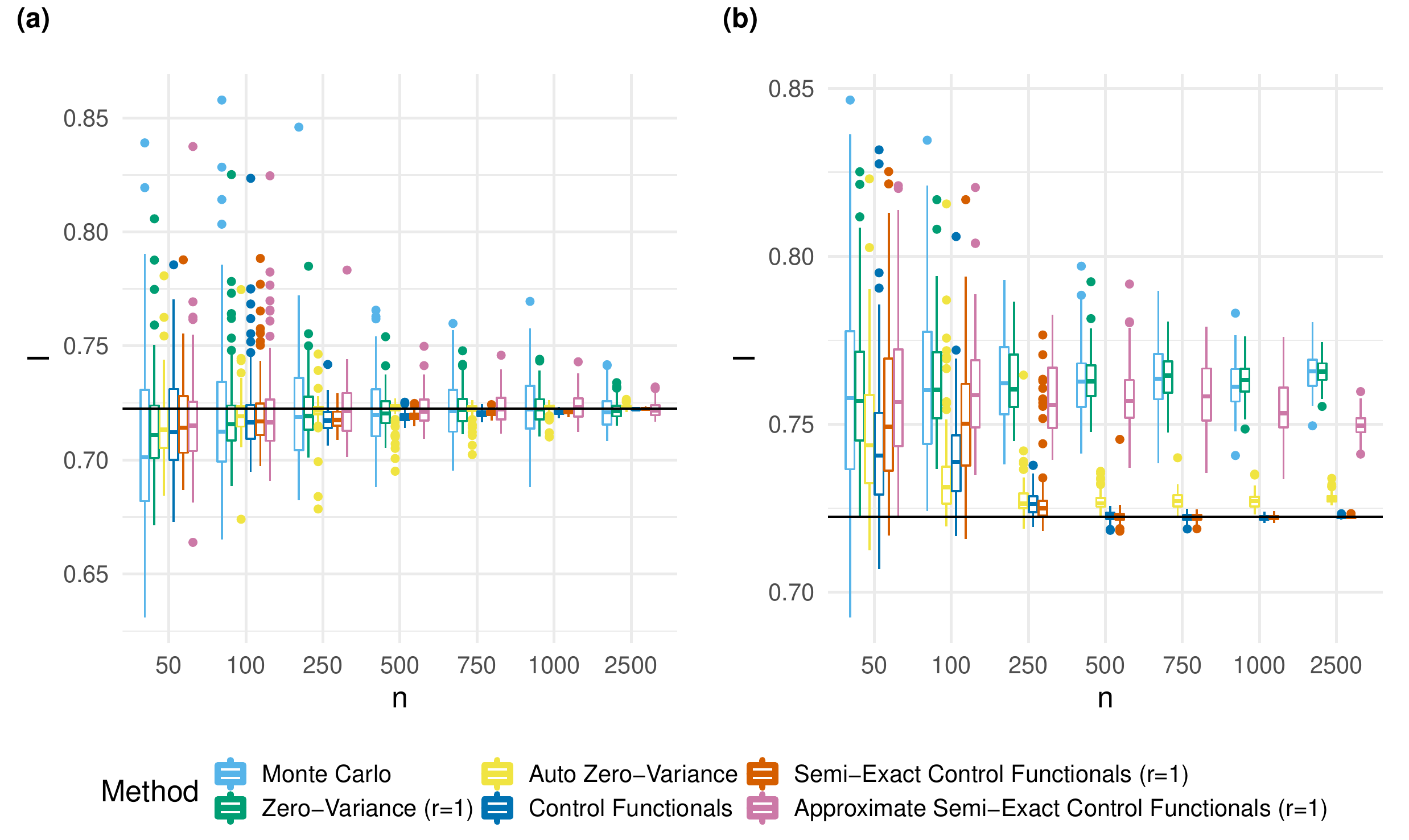}
 \caption{Recapture example (a) boxplots of 100 estimates of $\int x_1 P_{\b{x} \mid \b{y}} \mathrm{d} \b{x}$ when the Metropolis-adjusted Langevin algorithm is used for sampling and (b) boxplots of 100 estimates of $\int x_1 P_{\b{x} \mid \b{y}} \mathrm{d} \b{x}$ when the unadjusted Langevin algorithm is used for sampling. The black horizontal line represents the gold standard of approximation.
 }
    \label{fig:Recapture_marginal1}
\end{figure}

\FloatBarrier

\section{Reproducing Kernels and Worst-Case Error} \label{app: wce}

The purpose of this section is to review some basic results about  worst-case error analysis in a reproducing kernel Hilbert space context.
In Appendices~\ref{app: finite bound} and~\ref{ap: consistency proof} these results are used to prove Proposition~\ref{lem: KSD bound} and Theorem~\ref{thm: consistency}.

Let $k \colon \mathbb{R}^d \times \mathbb{R}^d \to \mathbb{R}$ be a positive-definite kernel such that $\int \lvert k(\b{x}, \b{y}) \rvert p(\b{x}) \dif \b{x} < \infty$ for every $\b{y} \in \mathbb{R}^d$ and $\mathcal{H}(k)$ the reproducing kernel Hilbert space of $k$.
The \emph{worst-case error} in $\mathcal{H}(k)$ of any weights $\b{v} = (v_1, \ldots, v_n) \in \mathbb{R}^n$ and any distinct points $\{\b{x}^{(i)}\}_{i=1}^n \subset \mathbb{R}^d$ is defined as
\begin{equation} \label{eq:wce-decomposition}
    e_{\mathcal{H}(k)}(\b{v} ; \{\b{x}^{(i)}\}_{i=1}^n) := \sup_{ \norm{h}_{\mathcal{H}(k)} \leq 1 } \, \bigg\rvert \int h(\b{x}) p(\b{x}) \dif \b{x} - \sum_{i=1}^n v_i h(\b{x}^{(i)}) \bigg\lvert.
\end{equation}
In this appendix we consider a fixed set of points $\{\b{x}^{(i)}\}_{i=1}^n$ and employ the shorthand $e_{\mathcal{H}(k)}(\b{v})$ for $e_{\mathcal{H}(k)}(\b{v} ; \{\b{x}^{(i)}\}_{i=1}^n)$.
Then a standard result \citep[see, for example, Section~10.2 in][]{NovakWozniakowski2010} is that the worst-case error admits a closed form
\begin{equation} \label{eq:wce}
e_{\mathcal{H}(k)}(\b{v}) = \bigg( \int \int k(\b{x}, \b{y}) p(\b{x}) p(\b{y}) \dif \b{x} \dif \b{y} - 2 \sum_{i=1}^n v_i \int k(\b{x}, \b{x}^{(i)}) p(\b{x}) \dif \b{x} + \b{v}^\top \b{K} \b{v} \bigg)^{1/2},
\end{equation}
where $\b{K}$ is the $n \times n$ matrix with entries $[\b{K}]_{i,j} = k(\b{x}^{(i)}, \b{x}^{(j)})$, and
\begin{equation} \label{eq:wce-decomposition}
    \bigg\lvert \int h(\b{x}) p(\b{x}) \dif \b{x} - \sum_{i=1}^n v_i h(\b{x}^{(i)}) \bigg\rvert \leq \norm{h}_{\mathcal{H}(k)} e_{\mathcal{H}(k)}(\b{v})
\end{equation}
for any $h \in \mathcal{H}(k)$.
Because the worst-case error in~\eqref{eq:wce} can be written as the quadratic form
\begin{equation*}
    e_{\mathcal{H}(k)}(\b{v}) = ( k_{pp} - 2\b{v}^\top \b{k}_p + \b{v}^\top \b{K} \b{v} )^{1/2},
\end{equation*}
where $k_{pp} =  \int \int k(\b{x}, \b{y}) p(\b{x}) p(\b{y}) \dif \b{x} \dif \b{y}$ and $[\b{k}_p]_i = \int k(\b{x}, \b{x}^{(i)}) p(\b{x}) \dif \b{x}$, the weights $\b{v}$ which minimise it take an explicit closed form:
\begin{equation*}
    \b{v}_\text{opt} = \argmin_{\b{v} \in \mathbb{R}^n} e_{\mathcal{H}(k)}(\b{v}) = \b{K}^{-1} \b{k}_p
\end{equation*}

Let $\Psi = \{ \psi_0, \ldots, \psi_{m-1} \}$ be a collection of $m \leq n$ basis functions for which the generalised Vandermonde matrix
\begin{equation*}
    \b{P}_\Psi = \begin{bmatrix} \psi_0(\b{x}^{(1)}) & \cdots & \psi_{m-1}(\b{x}^{(1)}) \\ \vdots & \ddots & \vdots \\ \psi_0(\b{x}^{(n)}) & \cdots & \psi_{m-1}(\b{x}^{(n)}) \end{bmatrix},
\end{equation*}
has full rank. In this paper we are interested in weights which satisfy the semi-exactness conditions $\sum_{i=1}^n v_i \psi(\b{x}^{(i)}) = \int \psi(\b{x}) p(\b{x}) \dif \b{x}$ for every $\psi \in \Psi$. 
Minimising the worst-case error under these constraints gives rise to the weights
\begin{equation} \label{eq:sard-weight-def}
    \b{v}_\text{opt}^\Psi = \argmin_{ \b{v} \in \mathbb{R}^n } e_{\mathcal{H}(k)}(\b{v}) \quad \text{ s.t. } \quad \sum_{i=1}^n v_i \psi(\b{x}^{(i)}) = \int \psi(\b{x}) p(\b{x}) \dif \b{x} \quad \text{ for every } \psi \in \Psi.
\end{equation}
These weights can be solved from the linear system~\citep[Theorem~2.7 and Remark~D.1]{Karvonen2018}
\begin{equation*}
    \begin{bmatrix} \b{K} & \b{P}_\Psi \\ \b{P}_\Psi^\top & \b{0} \end{bmatrix} \begin{bmatrix} \b{v}_\text{opt}^\Psi \\ \b{a} \end{bmatrix} = \begin{bmatrix} \b{k}_p \\ \b{\psi}_p \end{bmatrix},
\end{equation*}
where $\b{a} \in \mathbb{R}^q$ is a nuisance vector and the $i$th element of $\b{\psi}_p$ is $\int \psi_{i-1}(\b{x}) p(\b{x}) \dif \b{x}$. 
Note that~\eqref{eq:sard-weight-def} is merely a quadratic programming problem under the linear equality constraint $\b{P}_\Psi^\top \b{v} = \b{\psi}_p$.

These facts will be used in Appendices~\ref{app: finite bound} and~\ref{ap: consistency proof} to prove Proposition~\ref{lem: KSD bound} and Theorem~\ref{thm: consistency}.
Their relevance derives from the fact that $e_{\mathcal{H}(k_0)}(\b{v} ; \{\b{x}^{(i)}\}_{i=1}^n)$ coincides with the kernel Stein discrepancy between $p$ and the discrete measure $\sum_{i=1}^n v_i \delta(\b{x}^{(i)})$.

\section{Proof of Proposition \ref{lem: KSD bound}} \label{app: finite bound}

The following proof relies on the results reviewed in Appendix~\ref{app: wce}.

\begin{proof}[of Proposition \ref{lem: KSD bound}]
Applying the results reviewed in Appendix~\ref{app: wce} with $k = k_0$ and $\psi_j = \mathcal{L} \phi_j$, for which $\b{k}_p = \b{0}$ and $\b{\psi}_p = \b{e}_1$ from~\eqref{eq: kernel ints to 0}, we see that the solution to the optimisation problem
\begin{equation*}
    \b{v}_\text{opt}^\mathcal{F} = \argmin_{ \b{v} \in \mathbb{R}^n } e_{\mathcal{H}(k_0)}(\b{v} ; \{\b{x}^{(i)}\}_{i=1}^n) \quad \text{ s.t. } \quad \sum_{i=1}^n v_i h(\b{x}^{(i)}) = \int h(\b{x}) p(\b{x}) \dif \b{x} \quad \text{ for every } h \in \mathcal{F}
\end{equation*}
can be obtained by solving the linear system
\begin{equation*}
    \begin{bmatrix} \b{K}_0 & \b{P} \\ \b{P}^\top & \b{0} \end{bmatrix} \begin{bmatrix} \b{v}_\text{opt}^\mathcal{F} \\ \b{a} \end{bmatrix} = \begin{bmatrix} \b{0} \\ \b{e}_1 \end{bmatrix}.
\end{equation*}
A straightforward application of the block matrix inversion formula then gives
\begin{equation*}
    \b{v}_\text{opt}^\mathcal{F} = \b{K}_0^{-1} \b{P} ( \b{P}^\top \b{K}_0^{-1} \b{P} )^{-1} \b{e}_1 = \b{w} ,
\end{equation*}
where in the final equality we have recognised this expression as being identical to the weights $\b{w}$ used in our semi-exact control functional method, i.e. $I_\textup{SECF}(f) = \sum_{i=1}^n w_i f(\b{x}^{(i)})$ by~\eqref{eqn:BSS} and~\eqref{eq: weights}.
By~\eqref{eq: kernel ints to 0} the only non-zero element on the right-hand side of~\eqref{eq:wce} is $\b{v}^\top \b{K}_0 \b{v}$.
Thus we have characterised the weights $\bm{w}$ in the semi-exact control functional method as the solution to the problem
\begin{equation} \label{eq:integral-semi-exactness}
     \b{w} = \argmin_{ \b{v} \in \mathbb{R}^n } (\b{v}^\top \b{K}_0 \b{v})^{1/2} \quad \text{ s.t. } \quad \sum_{i=1}^n v_i h(\b{x}^{(i)}) = \int h(\b{x}) p(\b{x}) \dif \b{x} \quad \text{ for every } h \in \mathcal{F}.
\end{equation}
If $f = h + g$ with $h \in \mathcal{F}$ and $g \in \mathcal{H}(k_0)$, then it follows from the integral semi-exactness property~\eqref{eq:integral-semi-exactness} that
\begin{equation*}
    \lvert I(f) - I_{\text{SECF}}(f) \rvert = \lvert I(g) - I_{\text{SECF}}(g) + I(h) - I_{\text{SECF}}(h) \rvert = \lvert
    I(g) - I_{\text{SECF}}(g) \rvert.
\end{equation*}
Applying~\eqref{eq:wce-decomposition} and~\eqref{eq:integral-semi-exactness} yields
\begin{equation*}
    \lvert I(f) - I_{\text{SECF}}(f) \leq \norm{g}_{\mathcal{H}(k_0)} e_{\mathcal{H}(k_0)} (\b{w} ; \{\b{x}^{(i)}\}_{i=1}^n) = \norm{g}_{\mathcal{H}(k_0)} (\b{w}^\top \b{K}_0 \b{w})^{1/2}.
\end{equation*}
Since this bound is valid for any decomposition $f = h + g$ with $h \in \mathcal{F}$ and $g \in \mathcal{H}(k_0)$ we have
\begin{align*}
\lvert I(f) - I_{\text{SECF}}(f) \rvert \; \leq \; \inf_{\substack{f = h + g \\ h \in \mathcal{F}, \, g \in \mathcal{H}(k_0)}} \|g\|_{\mathcal{H}(k_0)} (\b{w}^\top \b{K}_0 \b{w})^{1/2} \; = \; |f|_{k_0,\mathcal{F}} \, (\b{w}^\top \b{K}_0 \b{w})^{1/2}
\end{align*}
as claimed.
\end{proof}

\section{Proof of Theorem \ref{thm: consistency}} \label{ap: consistency proof}

The following proof relies on the worst-case error results reviewed in Appendix~\ref{app: wce}, together with the following result, due to \citet{Hodgkinson2020}, which studies the convergence of the worst-case error (i.e. the kernel Stein discrepancy) of a weighted combination of the states $\{\b{x}^{(i)}\}_{i=1}^n$, where the weights $\tilde{\b{w}}$ are obtained by minimising the worst-case error subject to a non-negativity constraint:

\begin{theorem} \label{thm: hodgkinson}
Let $p$ be a probability density on $\mathbb{R}^d$ and $k_0 : \mathbb{R}^d \times \mathbb{R}^d \rightarrow \mathbb{R}$ a reproducing kernel which satisfies
\begin{equation*}
    \int k_0(\b{x}, \b{y}) p(\b{x}) \dif \b{x} = 0
\end{equation*}
for every $\b{y} \in \mathbb{R}^d$.
Let $q$ be a probability density with $p/q > 0$ on $\mathbb{R}^d$ and consider a $q$-invariant Markov chain $(\bm{x}^{(i)})_{i=1}^n$, assumed to be $V$-uniformly ergodic for some $V : \mathbb{R}^d \rightarrow [1,\infty)$, such that 
$$
\sup_{\b{x} \in \mathbb{R}^d}  \; V(\b{x})^{-r} \; \left( \frac{p(\bm{x})}{q(\bm{x})} \right)^4 \; k_0(\b{x},\b{x})^2   < \infty 
$$ 
for some $0 < r < 1$.
Let
\begin{equation} \label{eq: hodgkinson-weights}
    \tilde{\b{w}} = \argmin_{ \b{v} \in \mathbb{R}^n} e_{\mathcal{H}(k_0) }(\b{v} ; \{\b{x}^{(i)}\}_{i=1}^n) \quad \text{s.t.} \quad \sum_{i=1}^n v_i = 1 \quad \text{and} \quad \b{v} \geq \b{0}.
\end{equation}
Then  $e_{\mathcal{H}(k_0) }(\tilde{\b{w}} ; \{\b{x}^{(i)}\}_{i=1}^n) = O_P(n^{-1/2})$.
\end{theorem}
\begin{proof}
A special case of Theorem~1 in \citet{Hodgkinson2020}.
\end{proof}

\noindent The sense in which Theorem \ref{thm: hodgkinson} will be used is captured in the following corollary, which follows from the observation that removal of the non-negativity constraint in \eqref{eq: hodgkinson-weights} does not increase the worst-case error:

\begin{corollary} \label{cor: apply Hodge}
Under the same hypotheses as Theorem \ref{thm: hodgkinson}, let 
\begin{equation} 
    \bar{\b{w}} = \argmin_{ \b{v} \in \mathbb{R}^n} e_{\mathcal{H}(k_0) }(\b{v} ; \{\b{x}^{(i)}\}_{i=1}^n) \quad \text{s.t.} \quad \sum_{i=1}^n v_i = 1 .
\end{equation}
Then $e_{\mathcal{H}(k_0) }(\bar{\b{w}} ; \{\b{x}^{(i)}\}_{i=1}^n) \leq e_{\mathcal{H}(k_0) }(\tilde{\b{w}} ; \{\b{x}^{(i)}\}_{i=1}^n) = O_P(n^{-1/2})$.
\end{corollary}

Now the proof of Theorem \ref{thm: consistency} can be presented:

\begin{proof}[Proof of Theorem \ref{thm: consistency}]

From Assumption A2 and Lemma \ref{lem: comput etc} we have $(\bm{P}^\top \bm{K}_0^{-1} \bm{P})^{-1}$ is almost surely well-defined.
In the proof of Proposition~\ref{lem: KSD bound} we saw that
\begin{equation*}
    e_{\mathcal{H}(k_0)}( \b{w} ; \{\b{x}^{(i)}\}_{i=1}^n )^2 = \b{w}^\top \b{K}_0 \b{w}
\end{equation*}
and
\begin{align*}
    ( I(f) - I_{\textsc{SECF}}(f) )^2 & \leq |f|_{k_0,\mathcal{F}}^2 \; \b{w}^\top \b{K}_0 \b{w}
\end{align*}
Plugging in the expression for $\b{w} = \b{K}_0^{-1} \b{P} ( \b{P}^\top \b{K}_0^{-1} \b{P} )^{-1} \b{e}_1$ in \eqref{eq: weights}, we obtain
\begin{equation} \label{eq: wce-11-element}
    e_{\mathcal{H}(k_0)}( \b{w} ; \{\b{x}^{(i)}\}_{i=1}^n )^2 = [(\b{P}^\top \b{K}_0^{-1} \b{P})^{-1}]_{1,1}
\end{equation}
and
\begin{align*}
    ( I(f) - I_{\textsc{SECF}}(f) )^2 & \leq |f|_{k_0,\mathcal{F}}^2 [(\b{P}^\top \b{K}_0^{-1} \b{P})^{-1}]_{1,1} .
\end{align*}
It therefore suffices to consider the stochastic fluctuations of $[(\b{P}^\top \b{K}_0^{-1} \b{P})^{-1}]_{11}$ as $n \rightarrow \infty$.
To this end, let $[\bm{\Psi}]_{i,j} := \mathcal{L} \phi_j(\bm{x}^{(i)})$ and consider the block matrix
\begin{align}
\frac{\bm{P}^\top \bm{K}_0^{-1} \bm{P}}{\bm{1}^\top \bm{K}_0^{-1} \bm{1}} = \left[ \begin{array}{cc} 1 & \frac{\bm{1}^\top \bm{K}_0^{-1} \bm{\Psi}}{\bm{1}^\top \bm{K}_0^{-1} \bm{1}} \\
\frac{\bm{\Psi}^\top \bm{K}_0^{-1} \bm{1}}{\bm{1}^\top \bm{K}_0^{-1} \bm{1}} & \frac{\bm{\Psi}^\top \bm{K}_0^{-1} \bm{\Psi}}{\bm{1}^\top \bm{K}_0^{-1} \bm{1}} . \end{array} \right] \label{eq: troublesome matrix}
\end{align}
From the block matrix inversion formula we have
\begin{align}
\left[ \left( \frac{\bm{P}^\top \bm{K}_0^{-1} \bm{P}}{\bm{1}^\top \bm{K}_0^{-1} \bm{1}} \right)^{-1} \right]_{1,1} & = \left[ 1 -  \frac{\bm{1}^\top \bm{K}_0^{-1} \bm{\Psi} ( \bm{\Psi}^\top \bm{K}_0^{-1} \bm{\Psi} )^{-1} \bm{\Psi}^\top \bm{K}_0^{-1} \bm{1}}{\bm{1}^\top \bm{K}_0^{-1} \bm{1}}  \right]^{-1} \nonumber \\
& = \left[ 1 - \frac{\langle \bm{1} , \Pi \bm{1} \rangle_n }{\langle \bm{1} , \bm{1} \rangle_n } \right]^{-1} . \label{eq: big 11}
\end{align}
Since $\Pi = \bm{\Psi} (\bm{\Psi}^\top \bm{K}_0^{-1} \bm{\Psi})^{-1} \bm{\Psi}^\top \bm{K}_0^{-1}$ and
\begin{align*}
    \| \Pi \bm{1} \|_n^2 & = \bm{1}^\top \Pi^\top \bm{K}_0^{-1} \Pi \bm{1} \\
    & = \bm{1}^\top \bm{K}_0^{-1} \bm{\Psi} (\bm{\Psi}^\top \bm{K}_0^{-1} \bm{\Psi})^{-1} \bm{\Psi}^\top \bm{K}_{0}^{-1} \bm{\Psi} (\bm{\Psi}^\top \bm{K}_0^{-1} \bm{\Psi})^{-1} \bm{\Psi}^\top \bm{K}_0^{-1} \bm{1} \\
    & = \bm{1}^\top \bm{K}_0^{-1} \bm{\Psi} (\bm{\Psi}^\top \bm{K}_0^{-1} \bm{\Psi})^{-1} \bm{\Psi}^\top \bm{K}_0^{-1} \bm{1} \\
    & = \bm{1}^\top \bm{K}_0^{-1} \Pi \bm{1} \\
    & = \langle \bm{1} , \Pi \bm{1} \rangle_n ,
\end{align*}
our Assumption A3 implies that \eqref{eq: big 11} is almost surely asymptotically bounded, say by a constant $C \in [0,\infty)$.
In other words, it almost surely holds that
$$
[(\b{P}^\top \b{K}_0^{-1} \b{P})^{-1}]_{1,1} \; \leq \; C (\bm{1}^\top \bm{K}_0^{-1} \bm{1})^{-1} 
$$
for all sufficiently large $n$.

To complete the proof we evoke Corollary~\ref{cor: apply Hodge}, noting that the weights $\bar{\b{w}}$ defined in Corollary~\ref{cor: apply Hodge} satisfy $e_{\mathcal{H}(k_0) }(\bar{\b{w}} ; \{\b{x}^{(i)}\}_{i=1}^n) = (\b{1}^\top \b{K}_0^{-1} \b{1})^{-1/2}$, which follows from \eqref{eq: wce-11-element} with $\b{P}= \b{1}$.
Thus from Corollary~\ref{cor: apply Hodge} we conclude
\begin{equation*}
    [(\b{1}^\top \b{K}_0^{-1} \b{1})^{-1}]_{1,1}^{1/2} = e_{\mathcal{H}(k_0)}(\bar{\b{w}} ; \{\b{x}^{(i)}\}_{i=1}^n) \leq e_{\mathcal{H}(k_0)}(\tilde{\b{w}} ; \{\b{x}^{(i)}\}_{i=1}^n) = O_P(n^{-1/2}) ,
\end{equation*}
as required.

\end{proof}

\end{document}